\documentclass[11pt]{article}
\usepackage{amsfonts,amsmath,amssymb,amsthm,graphicx,url}
\usepackage{color}
\usepackage{algorithm}
\usepackage{algorithmic}
\usepackage{relsize}
\usepackage[margin=1in]{geometry}

\usepackage[T1]{fontenc}

\newcommand{\tabincell}[2]{\begin{tabular}{@{}#1@{}}#2\end{tabular}}
\usepackage{multirow}

\newtheorem{theorem}{Theorem}
\newtheorem{lemma}{Lemma}
\newtheorem{claim}{Claim}
\newtheorem{corollary}{Corollary}

\DeclareMathOperator*{\argmax}{arg\,max}

\newcommand{\bE}{\mathbb{E}}

\newcommand{\bR}{\mathbb{R}}

\newcommand{\cA}{{\cal A}}

\newcommand{\cD}{{\cal D}}

\newcommand{\cI}{{\cal I}}

\newcommand{\cM}{{\cal M}}

\newcommand{\des}{{\delta \epsilon}}

\begin{document}

\title{Bayesian Auctions with Efficient Queries\thanks{This work has been partially supported by 
NSF CAREER Award (No. 1553385), National NSF of China (No. 61741209) and the Fundamental Research Funds for the Central Universities. 
Part of this work was done when the first three authors were visiting Shanghai University of Finance and Economics. 
}}

\author{Jing Chen$^{\dag}$ \hspace{30pt} Bo Li$^{\dag}$ \hspace{30pt} Yingkai Li$^{\dag}$  \hspace{30pt} Pinyan Lu$^{\ddag}$\\
$^{\dag}$Department of Computer Science, Stony Brook University\\
Stony Brook, NY 11794, USA\\
\texttt{\{jingchen, boli2, yingkli\}@cs.stonybrook.edu}\\
$^{\ddag}$Institute for Theoretical Computer Science, Shanghai University of Finance and Economics\\
Shanghai 200433, China\\
\texttt{lu.pinyan@mail.shufe.edu.cn}}

\date{}

\maketitle

\begin{abstract}

Designing dominant-strategy incentive compatible (DSIC) mechanisms to generate (approximately)
optimal revenue is a fundamental problem in Bayesian mechanism design.
However, most existing studies 
assume that the seller knows the entire distribution. 
Unfortunately, this assumption may not hold in reality:
for example,
when the distributions have exponentially large supports or do not have succinct representations.
In this work we consider, for the first time, the {\em query complexity} of Bayesian mechanisms.
The seller only has limited oracle accesses to the players' distributions,
via {\em quantile queries}
and {\em value queries}.
We prove {\em logarithmic} lower-bounds for any constant approximation DSIC mechanisms in a large class of auction settings. 
For single-item auctions and multi-item auctions with unit-demand or additive valuation functions,
we prove {\em tight} upper-bounds via efficient query schemes. 
Thus, in those settings the seller needs to access much less than the entire distribution to achieve approximately optimal revenue.

\medskip

\noindent
{\bf Keywords:} mechanism design, the complexity of Bayesian mechanisms, query complexity, quantile queries, value queries

\end{abstract}

\thispagestyle{empty}

\setcounter{page}{0}

\newpage

\section{Introduction}\label{sec:intro}

An important problem in Bayesian mechanism design is to design auctions that (approximately)
maximize the seller's expected revenue.
More precisely, in a Bayesian multi-item auction a seller has $m$ heterogenous items to sell to $n$ players.
Each player $i$ has a private value for each item $j$, $v_{ij}$;
and each $v_{ij}$ is independently drawn from some prior distribution $\cD_{ij}$.
When the prior distribution $\cD\triangleq \times_{ij} \cD_{ij}$ is of {\em common knowledge} to both the seller and the players,
optimal Bayesian incentive-compatible (BIC) mechanisms have been discovered
for various auction settings \cite{myerson1981optimal, cremer1988full, cai2012algorithmic, cai2012optimal},
where all players reporting their true values forms a Bayesian Nash equilibrium.
When there is no common prior but the seller knows $\cD$,
many (approximately) optimal dominant-strategy incentive-compatible (DSIC) Bayesian mechanisms have been designed
\cite{myerson1981optimal, ronen2001approximating, chawla2010multi, kleinberg2012matroid, yao2015n, cai2016duality},
where it is each player's {\em dominant strategy} to report his true values.

However, the {\em complexity} for the seller to carry out such mechanisms is largely unconsidered in the literature.
Most existing Bayesian mechanisms require that the seller has full access to the prior distribution $\cD$
and is able to carry out all required optimizations based on $\cD$, so as to compute the allocation and the prices.
Unfortunately the seller may not be so knowledgeable or powerful in real-world scenarios.
If the supports of the distributions are exponentially large (in $m$ and $n$),
or if the distributions are continuous and do not have succinct representations,
it is hard for the seller to write out ``each single bit'' of the distributions or precisely carry out arbitrary optimization tasks based on them.
Even with a single player and a single item, when the value distribution is irregular, computing the
optimal price in time that is much smaller than the size of the support is not an easy task.
Thus, a natural and important question to ask is {\em how much the seller should know about the distributions in order to
obtain approximately optimal revenue.}

In this work we consider, for the first time, the {\it query complexity} of Bayesian mechanisms.
In particular, the seller can only access the distributions by making oracle queries.
Two natural types of queries are allowed, {\em quantile queries} and {\em value queries}.
That is, the seller queries the oracle with specific quantiles (respectively, values),
and the oracle returns the corresponding values (respectively, quantiles) in the underlying distributions.
These two types of queries happen a lot in market study.
Indeed, the seller may wish to know what is the price he should set so that half of the consumers would purchase his product;
or if he sets the price to be 200 dollars, how many consumers would buy it.
Another important scenario where such queries naturally come up is in databases.
Indeed, although the seller may not know the distribution,
some powerful institutes, say
the Office for National Statistics, may have such information figured out and stored in its database.
As in most database applications, it may be neither necessary nor feasible for the seller to
download the whole distribution to his local machines.
Rather, he would like to access the distribution via queries to the database.
Other types of queries are of course possible, and will be considered in future works.

In this work we focus on {\it non-adaptive} queries. That is,
the seller makes all oracle queries simultaneously,
before the auction starts. 
This is also natural in both database and market study scenarios,
and in Section \ref{sec:lower-bound}, 
we will show that the performance of adaptive queries cannot be improved up to a logarithmic factor.

\subsection{Main Results}

We would like to understand both lower- and upper-bounds for the query complexity of approximately optimal Bayesian auctions.
In this work, we mainly consider three widely studied settings: single-item auctions and
multi-item auctions with unit-demand or additive valuation functions.
Our main results are summarized in Table \ref{table:main results}.

Note that we allow arbitrary unbounded distributions that satisfy {\em small-tail assumptions},
with formal definitions provided in Section \ref{sec:small-tail}.
Similar assumptions are widely adopted in sampling mechanisms \cite{roughgarden2016ironing, devanur2016sample},
to deal with irregular distributions with unbounded supports.
Since distributions with bounded supports automatically satisfy the small-tail assumptions, the lower-bounds listed for the former apply to the latter as well.

\begin{table}[htbp]
\begin{center}
\renewcommand{\arraystretch}{1.4}
\begin{tabular}{|c|c|c|c|c|c|}
\hline
	~& Query & \multicolumn{4}{c|}{Distributions} \\
\cline{3-6}
	~&Complexity & \multicolumn{2}{c|}{Bounded in $[1,H]$} & \multicolumn{2}{c|}{Unbounded $\&$ Small Tail}  \\
\cline{1-6}
	\multirow4*{\rotatebox{90}{Auctions}} &Single-Item & \multicolumn{2}{c|}{$\Theta(n\epsilon^{-1}\log H)$}
	& \multicolumn{2}{c|}{$O(-n\epsilon^{-1}\log h(\frac{2\epsilon}{3(1+\epsilon)})$} \\
\cline{2-6}
	~&Unit-Demand & $\forall c>1$: $\Omega(\frac{mn\log H}{\log c})$ & $\forall c>24$: $O(\frac{mn\log H}{\log (c/24)})$
	& \multicolumn{2}{c|}{$\forall c>24$: $O(-\frac{mn\log h(\frac{2c-48}{3c})}{\log (c/24)})$} \\
\cline{2-6}
	~&Additive & $\forall c>1$: $\Omega(\frac{mn\log H}{\log c})$ &
	$\forall c>8$: $O(\frac{mn\log H}{\log (c/8)})$
	& \multicolumn{2}{c|}{$\forall c>8$:  $O(-\frac{m^2n\log h(\frac{c-8}{10c})}{\log (c/8)})$}\\
\cline{2-6}
	~& \multirow{1}*{Single-Item} & \multicolumn{4}{c|}{Regular Distributions:
$\Omega(n\epsilon^{-1})$, $O(n\epsilon^{-1}\log\frac{n}{\epsilon})$} \\
\hline
\end{tabular}
\caption{Our main results. Here $h(\cdot)<1$ is the tail function in the small-tail assumptions.
For single-item auctions, the revenue is
a $(1+\epsilon)$-approximation to the optimal BIC revenue, with $\epsilon$ sufficiently small.
For multi-item auctions with unit-demand or additive valuation functions, the revenue is a $c$-approximation for some constant $c$.
}
\label{table:main results}
\end{center}
\end{table}

Also note that our lower- and upper-bounds on query complexity are {\em tight} for bounded distributions.
As will become clear in Section \ref{sec:lower-bound},
our lower-bounds allow the seller to make both value and quantile queries,
and apply to any multi-player multi-item auctions
 where
each player's valuation function is {\em succinct sub-additive}:
formal definitions in Section \ref{sec:lower-bound}.
The lower-bounds also allow randomized queries and randomized mechanisms.

For the upper-bounds,
all our query schemes are deterministic and only make one type of queries:
value queries
for bounded distributions and
quantile queries in the other cases; see Sections~\ref{sec:bounded} and \ref{sec:queryunbounded}.
We show that our schemes, despite of being very efficient, only loses a small fraction of revenue compared with when the seller has full access to
the distributions.



\subsection{Discussions}

\subparagraph*{Sample Complexity.}
A closely related area to our work is sampling mechanisms
\cite{cole2014sample,azar2014prophet, huang2015making, dhangwatnotai2015revenue, morgenstern2016learning,
devanur2016sample, goldner2016prior, gonczarowski2017efficient,cai2017learning}.
It assumes that the seller does not know $\cD$
but observes independent samples from $\cD$ before the auction begins.
The {\em sample complexity} measures
how many samples the seller needs so as to obtain a good approximation to the optimal Bayesian revenue.
The best-known sample complexity results are summarized in Table \ref{table:sample:laterature}.
\begin{table}[htbp]
\begin{center}
\begin{tabular}{|c|c|c|c|c|}
\hline
Auctions &
{\tabincell{c}{Single-Item \\(regular) \cite{devanur2016sample}}} &
{\tabincell{c}{Single-Item \\ (bounded) \cite{devanur2016sample}}}
& {\tabincell{c}{Unit-Demand \\(bounded)  \cite{morgenstern2016learning}}} &
{\tabincell{c}{Additive \\(bounded)\cite{cai2017learning}}}  \\
\hline
{\tabincell{c}{Sample \\ Complexity}}
& $\tilde{O}(n\epsilon^{-4})$
& $\tilde{O}(nH\epsilon^{-3})$
& $\tilde{O}(nm^{2}H^{2}\epsilon^{-2})$
& $\tilde{O}(nm^{2}H^{2}\epsilon^{-2})$ \\
\hline
Approximations & $1+\epsilon$ & $1+\epsilon$ & 27 & 32 \\
\hline
\end{tabular}
\end{center}
\caption{Sample Complexity. When a distribution is bounded, it is bounded within $[1,H]$.
For multi-item auctions with unit-demand or additive valuation functions,
the revenue has an extra $\epsilon$ additive loss.
}
\label{table:sample:laterature}

\end{table}%

Oracle queries can be seen as {\em targeted samples},
where the seller actively asks the information he needs rather than passively
learns about it from random samples.
As such, it is intuitive that queries are more efficient than samples,
but it is a priori unclear how efficient queries can be.
Our results answer this question quantitatively and show that query complexity can be exponentially smaller than
sample complexity: the former is {\em logarithmic} in the ``size'' of the distributions,
while the latter is polynomial.

Finally, the design of query mechanisms facilitates the design of sampling mechanisms.
If the seller observes enough samples from $\cD$,
then he can mimic quantile queries and apply query mechanisms: see Section \ref{sec:app:sample} for more details.

\subparagraph*{Parametric Auctions.}
Parametric mechanisms \cite{azar2013parametric,azar2013optimal}
assume the seller only knows some specific parameters about the distributions, such as the mean, the median (or a single quantile), and the variance.
Note that using quantile or value queries,
one can get the exact value of the median and the approximate value of the mean, and then apply parametric mechanisms.
However, existing parametric mechanisms only consider single-parameter auctions, where the distributions are regular or have monotone hazard rate.
Since our mechanisms make non-adaptive oracle queries,
our results imply parametric mechanisms in multi-parameter settings with general distributions,
where the ``parameters'' are the oracle's answers to our query schemes.
Our lower-bounds imply that knowing only the median is not enough to achieve the same approximation ratios as we do.
Finally, it remains unknown whether constant approximations can be achieved for multi-parameter auctions or general distributions, knowing only the mean and the variance.


\subparagraph*{Distributions within Bounded Distance.}
Recently, \cite{cai2017learning} studies auctions where the true prior distribution is unknown to the seller,
but he is given a distribution that is close to the true prior, as measured by the Kolmogorov distance.
On the one hand, the learnt distributions from our query schemes can be far from the true prior in terms of the Kolmogorov distance, thus their mechanisms do not apply.
On the other hand, although a distribution close to the true prior may be learnt via sufficiently many oracle queries, our lower-bounds
imply that the query complexity of this approach will not be better than ours.

\subparagraph*{Using Experts as Oracles.}
If the players' value distributions are known by some experts, then the seller can use
the experts as oracles.
Indeed, we
are able to design proper scoring rules \cite{brier1950verification, cervera1996proper}
for the seller to elicit truthful answers from the experts for his queries.
If the experts are actually players in the auctions, then they have their own stakes about the final allocation and prices,
and it would be interesting to see how the seller can still use them as oracles and
get truthful answers for his queries, while keeping them truthful about their own values.
See \cite{jing2016crowd} for more discussions on this front.

\subsection{Other Related Works}

The complexity of auctions is an important topic
in the literature, and several complexity measures have been considered.
Following the taxation principle \cite{hammond1979straightforward, guesnerie1981taxation},
\cite{hart2013menu} defines the {\em menu complexity} of truthful auctions.
For a single additive buyer,
\cite{daskalakis2013mechanism} shows
the optimal Bayesian auction for revenue can have an infinite menu size or a continuum of menu entries,
and \cite{babaioff2014simple}
shows a constant approximation
under finite menu complexity.
Recently, \cite{dobzinski2016computational}
considers the taxation, communication, query and menu complexities of
truthful combinatorial auctions,
and shows important connections among them.
The queries considered there are totally different
from ours: we are concerned with the complexity of accessing
the players' value distributions in Bayesian settings, while \cite{dobzinski2016computational} is concerned with
the complexity of accessing the players' valuation functions in non-Bayesian settings.

\subsection{Future Directions}

Many interesting questions about the query complexity of Bayesian auctions are worth exploring.
First, as mentioned, 
there is a logarithmic gap between adaptive queries and non-adaptive queries for multi-item auctions. 
It is intriguing to design approximately optimal Bayesian mechanisms with matching query complexity using adaptive queries. 
Another interesting direction is when the answers of the oracle contain noise.
In this case,
the distributions learnt by the seller may be within a small distance from
the ``true distributions'' defined by oracle answers without noise.
This is related to \cite{cai2017learning} and it would be interesting to design mechanisms to handle such noise.

\section{Preliminaries}\label{sec:pre}

\subsection{Bayesian Auctions}
\label{sec:prelim:bayesian}

In a multi-item auction there are $m$ items, denoted by $M = \{1, \dots, m\}$,
and $n$ players, denoted by $N = \{1,\dots, n\}$.
Each player $i\in N$ has a non-negative value for each item $j\in M$, $v_{ij}$,
which is independently drawn from distribution $\cD_{ij}$.
Player $i$'s {\em true valuation} is $(v_{ij})_{j\in [m]}$.
To simplify the notations, we may write $v_{i}$ for $(v_{ij})_{j\in [m]}$ and $v$ for $(v_{i})_{i\in [n]}$.
Letting $\cD_{i}=\times_{j\in M} \cD_{ij}$ and $\cD= \times_{i\in N} \cD_{i}$,
we use $\cI = (N, M, \cD)$ to denote the corresponding Bayesian auction instance
and $OPT(\cI)$ the optimal BIC revenue of $\cI$.
When $\cI$ is clear from the context, we write $OPT$ for short.

We will consider several classes of widely studied auctions.
A {\em single-item} auction has $m=1$.
When $m>1$, a bidder $i$ being {\em unit-demand} means his
value for a subset $S$ of items is $\max_{j\in S}v_{ij}$,
and a bidder $i$ being {\em additive} means his value for $S$ is $\sum_{j\in S} v_{ij}$.
When all bidders are unit-demand (respectively, additive),
we call such an auction a {\em unit-demand auction} (respectively, an {\em additive auction}) for short.

\subsection{Query Complexity}
\label{sec:query:oracle}

In this work, we only allow the seller to access the prior distributions via two types of oracle queries:
{\em value queries} and {\em quantile queries}.
Given a distribution $D$ over reals,
in a value query, the seller sends a value $v\in \bR$ and the oracle returns
the corresponding quantile $q(v) \triangleq \Pr_{x\sim D}[x \geq v]$.
In a quantile query, the seller sends a quantile $q\in [0, 1]$
and the oracle returns the corresponding value $v(q)$
such that $\Pr_{x\sim D}[x \geq v(q)] = q$.
With {\em non-adaptive} queries, the seller first sends all his queries to the oracle, gets the answers back, and then
runs the auction. 
The {\em query complexity}
is the number of queries made by the seller.

Note that the answer to a value query is unique.
The quantile queries are a bit tricky, as
for discrete distributions
there may be multiple values corresponding to the same quantile~$q$,
or there may be none.
When there are multiple values,
to resolve the ambiguity,
let the output of the oracle be the largest one: that is,
$v(q) = \argmax_{z}\{ \Pr_{x\sim D}[x \geq z] = q\}$.
When there is no value corresponding to~$q$,
the oracle returns
the largest value whose corresponding quantile is larger than $q$: that is,
$v(q) = \argmax_{z} \{\Pr_{x\sim D}[x \geq z] > q\}$.
So for any quantile query $q$,
$v(q) = \argmax_z\{\Pr_{x\sim D}[x \geq z] \geq q\}$ in general.
For any discrete distribution $D$ and quantile query $q>0$,
$v(q)$ is always in the support of $D$. When $q=0$, $v(q)$ may be $+\infty$.

\section{Lower Bounds}
\label{sec:lower-bound}

In this section, we prove lower bounds for the query complexity of Bayesian mechanisms, and we focus on DSIC mechanisms.  
As a building block for our general lower bound,
we first have the following for {\em single-item single-player} auctions. 
Note that although the main focus of this paper is about non-adaptive queries,
here we prove a stronger result by showing that even with adaptive queries, 
our bounds for {\em single-item single-player} auctions are still tight.
%
Formally, we have the following lemma.

\begin{lemma}\label{lem:adaptive:single}
For any constant $c > 1$, there exists a constant $C$ such that,
for any large enough $H$,
any Bayesian mechanism $\cM$
making less than $C \log_c H$
adaptive value and quantile queries to the oracle,
there exists a single-player single-item Bayesian auction instance $\cI = (N, M, \cD)$
where the values are bounded in $[1,H]$, such that
$Rev(\cM(\cI))<
\frac{OPT(\cI)}{c}$.
\end{lemma}

\begin{proof}
Here we consider the equal revenue curve,
i.e., $F(v) = 1-\frac{1}{v}, v\in [1,H)$ and $F(H) = 1$.
For this distribution, posting any price from $[1,H]$ has exactly the same revenue 1.
For any constant $H$, let $k \triangleq \lfloor \frac{1}{4} \log_{(4c)^{4c+2}} H\rfloor$.
We divide the value interval $[1, H]$
into $k+1$ sub-intervals as follows:
from right to left,
$u_{k} = H$, and $u_s = \frac{u_{s+1}}{(4c)^{4c+2}}$ for each $s\in \{k-1, \dots, 0\}$.
For each value interval $(u_s, u_{s+1})$, there exists a corresponding quantile interval $(q_s, q_{s+1}) \triangleq (\frac{1}{u_s}, \frac{1}{u_{s+1}})$.


For each value $s\in \{k-1, \dots, 0\}$, consider pair of interval $(u_s, u_{s+1})$ and $(q_s, q_{s+1})$,
we construct $[\lceil 4c \rceil]$ distributions for it.
More precisely, the distribution $\cD^s_z$ for each $z\in[\lceil 4c \rceil]$
is defined as follows.
For $v_z < u_s$ or $u_{s+1} < v_z < H$,
the density of $v_{z}$ is $\frac{1}{v^2_z}$.
The probability of $v_z = H$ is $\frac{1}{H}$.
The probability of $(4c)^z u_s$ is $q_{s+1}-q_s$.

By construction, given parameter $s$,
all $[\lceil 4c \rceil]$ distributions coincide outside the value and quantile range $(u_s, u_{s+1})$ and $(q_s, q_{s+1})$.
Therefore, only a query inside the range can distinguish those distributions from each other.
Moreover, for any two parameter $s, s' \in \{k-1, \dots, 0\}$, the value and quantile queries will get the exactly same response for queries outside their own value and quantile intervals.
Thus, for any adaptive value and quantile queries, a single query can distinguish at most one set of $[\lceil 4c \rceil]$ distributions.
Letting $c' \triangleq 1-\frac{1}{2c}$ and
$C \triangleq  \frac{1-c'}{8(4c+2)\log_c (4c)} = \frac{1}{16c(4c+2)\log_c (4c)}$,
we have $C \log_c H < k (1-c')$.
Accordingly, for any Bayesian mechanism $\cM$ that makes less than
$C \log_c H$ adaptive value and quantile queries, 
there exists a value $s$ such that,
with probability at least $c'$,
$\cM$ cannot distinguish $\cD^s_z$'s from each other.

We now analyze the optimal BIC revenue for those instances.
For any $\cI_z = (N, M, \cD_z)$,
Myerson's mechanism is optimal:
it sets a (randomized) threshold for the unique player,
if the player bids at least the threshold then he gets the item and pays the threshold payment,
otherwise the item is unsold.
It is not hard to verify that
$OPT(\cI_z) = (4c)^z u_s q_{s+1}$ for each $\cI_z$.

Next, we analyze the revenue of $\cM$.
Since $\cM$ is DSIC, the allocation rule must be monotone in the player's bid,
and he will pay the threshold payment set by $\cM$, denoted by $P$.
Here $P$ may also be randomized.
Note that for all instances,
setting $P<4cu_s$
is strictly worse than setting $P = 4cu_s$,
and setting $P>(4c)^{\lceil 4c \rceil} u_s$
is strictly worse than setting $P = (4c)^{\lceil 4c \rceil} u_s< u_{s+1}$.
Also, for any instance $\cI_z$ and
any $z'\in \{1, \dots, \lceil 4c \rceil - 1\}$,
setting $P\in ((4c)^{z'} u_s, (4c)^{z'+1} u_s)$
is strictly worse than setting $P = (4c)^{z'+1} u_s$.
Thus, when mechanism $\cM$ cannot distinguish the $\cI_z$'s,
it must use the same $P$ for all $\cI_z$'s, and
the best it can do is to set $P = (4c)^z u_s$ with some probability $\rho_z$
for each $z\in [\lceil 4c \rceil]$.
Because $\sum_{z \in [\lceil 4c \rceil]} \rho_z = 1$,
there exists $z^*$ such that $\rho_{z^*} \leq \frac{1}{4c}$.
Thus we have
\begin{eqnarray*}
Rev(\cM(\cI_{z^*})) &\leq& \frac{1}{4c} \cdot (4c)^{z^*} \cdot u_s \cdot q_{s+1}
+ (1-\frac{1}{4c}) (4c)^{z^*-1} \cdot u_s \cdot q_{s+1} \\
&<& \frac{1}{2c} \cdot (4c)^{z^*} \cdot u_s \cdot q_{s+1}
= \frac{1}{2c} OPT(\cI_{z^*}),
\end{eqnarray*}
where the first inequality is because for any threshold other than $(4c)^{z^*}u_s$,
the resulting expected revenue is no larger than that with the threshold being
$(4c)^{z^*-1}u_s$.
That is, when $\cM$ cannot distinguish the $\cI_z$'s, it cannot be a $2c$-approximation for $\cI_{z^*}$.

As the revenue of $\cM$ under $\cI_{z^*}$
is at most $OPT(\cI_{z^*})$ when it is able to distinguish $\cI_{z^*}$ from all the other instances,
we have
$$Rev(\cM(\cI_{z^*}))
\leq (1-\frac{1}{2c})\frac{1}{2c}OPT(\cI_{z^*}) + \frac{1}{2c}OPT(\cI_{z^*})
< \frac{1}{c}OPT(\cI_{z^*}).$$
Thus $\cM$ is not a $c$-approximation for $\cI_{z^*}$, and Lemma \ref{lem:adaptive:single} holds.
\end{proof}

We extend this lemma to arbitrary
multi-player multi-item Bayesian auctions with {\em succinct sub-additive} valuations,
as follows. 
To make our exposition clearer, we first introduce some notations.
A very broad class of Bayesian auctions, {\em (monotone) sub-additive} auctions,
is such that each player~$i$ has a valuation function
$v_i: 2^{[m]}\rightarrow \bR$, which
satisfies $v_i(S) + v_i(T) \geq v_i(S \cup T) \geq v_i(S) \geq 0$ for any subsets of items $S$ and~$T$.
As such a valuation function in general needs $2^m$ values to describe,
following the conventions in Bayesian auction design
\cite{rubinstein2015simple, chawla2016mechanism, cai2017simple},
we will consider {\em succinct sub-additive} auctions, where
only the item-values, that is, the $v_{ij}$'s, are independently drawn from
the underlying distribution $\cD = \times_{i\in[n], j\in [m]} \cD_{ij}$.
Given $(v_{ij})_{j\in [m]}$, it is publicly known how to compute player $i$'s value for any subset of items.
That is, the valuation function $v_i$ now takes a vector of item-values $(v_{ij})_{j\in [m]}$ and a subset $S\subseteq [m]$ as inputs,
such that for any vector $(v_{ij})_{j\in [m]}$,
the resulting function $v_i((v_{ij})_{j\in [m]}, \cdot)$ is sub-additive and
$v_i((v_{ij})_{j\in [m]}, \{j\}) = v_{ij}$ for each item $j$.
Note that such auctions include single-item, unit-demand
and
additive auctions as special cases.

\begin{theorem}\label{thm:adaptive:multi}
For any constant $c > 1$,
there exists a constant $C$ such that, for any $n\geq 1, m\geq 1$,
for any large enough $H$,
any monotone sub-additive valuation function profile $v = (v_i)_{i\in [n]}$,
and any Bayesian mechanism $\cM$ making less than $C nm \log_{m \ln H} H$
adaptive value and quantile queries to the oracle,
there exists a multi-item Bayesian auction instance $\cI = (N, M, \cD)$,
where $|N|=n, |M|=m$ and the values are bounded in $[1,H]$, such that
$Rev(\cM(\cI)) < \frac{OPT(\cI)}{c}$.
\end{theorem}
\begin{proof}

Similar to the proof of Lemma \ref{lem:adaptive:single},
for any $H$, let $k \triangleq \lfloor \frac{1}{4} \log_{(4cx)^{4c+2}} H\rfloor$,
where $x$ is a parameter to be determined later.
Let $C' \triangleq \frac{1}{24c(4c+2)\log_c (4cx)}$
and $H$ be large enough so that $k\geq 1$.
We divide the value interval $[1, H]$
into $k+1$ sub-intervals as follows:
from right to left,
$u_{k} = H$, and $u_s = \frac{u_{s+1}}{(4cx)^{4c+2}}$ for each $s\in \{k-1, \dots, 0\}$.
For each value interval $(u_s, u_{s+1})$, there exists a corresponding quantile interval $(q_s, q_{s+1}) \triangleq (\frac{1}{u_s}, \frac{1}{u_{s+1}})$.

It is easy to see that $C'nm\log_c H < \frac{nmk}{3c}$. 
Thus, for any Bayesian mechanism $\cM$
making less than $C'nm\log_c H$ adaptive value and quantile queries,
with probability at least~$1-\frac{1}{3c}$,
there exists a player-item pair $(i^{*},j^{*})$,
a value interval $(u_s, u_{s+1})$ with corresponding
quantile interval $(q_s, q_{s+1})$ such that
$\cM$ does not query these two intervals for $\cD_{i^* j^*}$
and does not distinguishes $\cD_{i^*j^*;s}^z$'s from each other.

Therefore, for each $i\in [n], j\in [m], s\in [k]$, 
we construct $\lceil 4c \rceil$ Bayesian instances $\{\cI^z = (N, M, \cD^z)\}_{z\in[\lceil 4c \rceil]}$,
where each $\cD^z_{ij}$ is equal revenue distribution if $i \neq i^{*}$ or $j \neq j^{*}$.
For any $z\in [\lceil 4c \rceil]$,
we construct the distribution $\cD_{ij;s}^z$ such that
for $v_z < u_s$ or $u_{s+1} < v_z < H$,
the density of $v_{z}$ is $\frac{1}{v^2_z}$.
The probability of $v_z = H$ is $\frac{1}{H}$.
The probability of $(4cx)^z u_s$ is $q_{s+1}-q_s$.
All distributions other than $\cD_{ij}$ are equal revenue distribution. 
For any queries outside the interval 
$(u_s, u_{s+1})$ and $(q_s, q_{s+1})$, it cannot distinguish those instances. 
Thus, with less than $Cnm\log_c H$ adaptive value and quantile queries,
with probability at least~$1-\frac{1}{3c}$, 
the mechanism cannot distinguish those instances from each other.

We now analyze the optimal BIC revenue for those instances.
For any $\cI_z$,
Myerson's mechanism is optimal:
it sets a (randomized) threshold for the unique player,
if the player bids at least the threshold then he gets the item and pays the threshold payment,
otherwise the item is unsold.
Letting $\delta \triangleq \frac{1}{H}$, it is not hard to verify that
$OPT(\cI_z) = (4cx)^z u_s q_s$ for each $\cI_z$.

Given any succinct sub-additive valuation function profile $v = (v_i)_{i\in [n]}$ where each
$v_i$ takes a vector of
item-values $(v_{ij})_{j\in [m]}$
as part of its input,
we would like to compare the optimal revenue for the sub-additive instances
defined by the $\cI^z$'s with the corresponding expected revenue of~$\cM$.
By construction, the $\cD^z$'s
differ only at the $\cD^z_{i^*j^*}$'s, within the value interval
$(u_s, u_{s+1})$ and the
quantile interval $(q_s, q_{s+1})$.
Accordingly, with probability at least $1-\frac{1}{3c}$,
mechanism $\cM$ cannot distinguish the $\cI^z$'s from each other.
Eventually, we will analyze the revenue of $\cM$ conditional on this event happening.

For now, to compare the optimal revenue and that of $\cM$, let us first introduce some notations.
For any item-value profile
$\hat{v} = (\hat{v}_{ij})_{i\in [n], j\in [m]}$,
when the players bid $\hat{v}$,
we denote by $x_i(\hat{v})$
the (randomized) allocation
of $\cM$ to a player $i$.
It is defined by the probabilities $\sigma_{iS}(\hat{v})$ for
all the subsets $S\subseteq [m]$:
each $\sigma_{iS}(\hat{v})$
is the probability that player $i$ receives $S$ under bid $\hat{v}$.
Accordingly, the expected value of player $i$ for allocation $x_{i}(\hat{v})$
is
$v_i((\hat{v}_{ij})_{j\in [m]}, x_{i}(\hat{v})) =
\sum_{S}v_i((\hat{v}_{ij})_{j\in [m]}, S)\cdot \sigma_{iS}(\hat{v})$.
Moreover, for each item $j$,
let $x_{ij}(\hat{v})$ be the probability that player $i$ receives item $j$ according to $x_i(\hat{v})$:
that is,
$x_{ij}(\hat{v}) = \sum_{S: j\in S}\sigma_{iS}(\hat{v})$.

We upper-bound the revenue of $\cM$ in three steps.
To begin with, we reduce the multi-player sub-additive instances to {\em single-player} sub-additive instances,
and construct a DSIC Bayesian mechanism $\cM^*$
that only sells the items to player~$i^*$.
Given
any instance $\cI^z$,
mechanism $\cM^*$ runs on the single-player sub-additive instance $\cI^z_{i^*} = (\{i^*\}, M, \cD^z_{i^*})$.
It first simulates the item values of players in $N \setminus \{i^*\}$, which are all 1's, and
then runs $\cM$.
Mechanism $\cM^*$ answers the oracle queries of $\cM$ truthfully.
The allocation and the payment for player $i^*$ under $\cM^*$
is the same as those under $\cM$.
For any player $i\neq i^*$, mechanism $\cM^*$ assigns nothing to him and charges him  $0$,
because $i$ is an imaginary player to $\cM^*$.
It is easy to see that mechanism $\cM^*$ is DSIC.
Moreover,
\begin{equation}\label{eq:cm1}
Rev(\cM^*(\cI^z_{i^*})) \geq Rev(\cM(\cI^z))
- \bE_{\hat{v} \sim \cD^z} \sum_{i\neq i^*}
v_i((\hat{v}_{ij})_{j\in [m]}, x_{i}(\hat{v})),
\end{equation}
because the revenue generated by $\cM$ from players in $N\setminus\{i^*\}$
is at most their total value for the allocation.

Next, we reduce the single-player sub-additive instances to single-player {\em additive} instances,
and
construct a
DSIC Bayesian mechanism
$\cM^+$
that runs on the single-player
additive instances $\cI^{+z}_{i^*} = (\{i^*\}, M, \cD^z_{i^*})$,  with $z\in [\lceil 4c \rceil]$.
Note that each $\cI^{+z}_{i^*}$ has the same item-value distributions as $\cI^z_{i^*}$, but
player $i^*$'s value for any subset of items is additive.

For each single-player sub-additive instance defined by $\cI^z_{i^*}$
and the valuation function profile $v$,
by the taxation principle~\cite{guesnerie1981taxation},
mechanism $\cM^*$ is equivalent to providing
a menu of options to player $i^*$
and then
letting $i^*$ choose a menu entry maximizing his expected utility according to his true valuation. 
Given any instance $\cI^{+z}_{i^*}$, mechanism $\cM^+$
provides the same menu as mechanism $\cM^*$ under $\cI^{z}_{i^*}$ and $v$,
 except that the payment in each entry is discounted by a multiplicative $1-\hat{\epsilon}$.
Here $\hat{\epsilon}$ is a sufficiently small constant in $(0, 1)$ to be determined later in the analysis.
The truthfulness of $\cM^+$ is immediate, because it lets~$i^*$ choose a menu entry maximizing his expected utility under
his true additive values.
Let
$$\bar{\delta} \triangleq \bE_{\hat{v}_{i^*} \sim \cD^z_{i^*}}
\max_{S \subseteq [m]}
(\sum_{j \in S} \hat{v}_{i^* j} - v_{i^*}((\hat{v}_{i^*j})_{j\in [m]}, S)),$$
the expected maximum difference between the additive values and the succinct sub-additive values.
Following Lemma 3.4 in \cite{rubinstein2015simple}, which compares the revenue in the sub-additive instance with that in the
corresponding additive instance, we have
\begin{equation}\label{eq:cm2}
Rev(\cM^+(\cI^{+z}_{i^*})) \geq (1-\hat{\epsilon})(Rev(\cM^*(\cI^z_{i^*})) - \bar{\delta} / \hat{\epsilon}).
\end{equation}

Finally, we reduce the single-player additive instances to single-player {\em single-item} instances, and
construct a DSIC Bayesian mechanism $\cM'$ that only sells item $j^{*}$ to player $i^*$.
Mechanism~$\cM'$
runs on the single-player single-item instances
$\cI^z_{i^*j^*} = (\{i^*\}, \{j^*\}, \cD^z_{i^*j^*})$, with $z\in [\lceil 4c \rceil]$.
Given any~$\cI^z_{i^*j^*}$, it first lets player $i^*$ report $\hat{v}_{i^*j^*}$.
Then it simulates the $\hat{v}_{i^* j}$'s from $\cD^z_{i^*j}$ for $j\neq j^*$, which are all~1's,
and runs $\cM^+$ on the
augmented additive instance $\cI^{+z}_{i^*}$ to obtain allocation $x^+_{i^*}(\hat{v}_{i^*})$ and payment
$p^+_{i^*}(\hat{v}_{i^*})$.
For each item $j$, let $x^+_{i^*j}(\hat{v}_{i^*})$ be the probability that
player $i^*$ receives item $j$ in the allocation.
Mechanism $\cM'$ sets its outcome to be the following:
\begin{itemize}
\item $x'_{i^*j^{*}}(\hat{v}_{i^*j^*}) = x^+_{i^*j^{*}}(\hat{v}_{i^*})$; and
\item $p'_{i^*}(\hat{v}_{i^*j^*}) = p^+_{i^*}(\hat{v}_{i^*}) - \sum_{j \in [m] \setminus \{j^*\}}\hat{v}_{i^*j}x^+_{i^*j}(\hat{v}_{i^*})$.
\end{itemize}

\noindent
Note that $p'_{i^*}(\hat{v}_{i^*j^*})$ may be negative.
By Lemma 21 of \cite{hart2012approximate}, mechanism $\cM'$ is DSIC and
\begin{equation}\label{eq:cm3}
Rev(\cM'(\cI^z_{i^*j^*})) \geq Rev(\cM^+(\cI^{+z}_{i^*}))
- \sum_{j\neq j^*} \bE_{\hat{v}_{i^*j} \sim \cD^z_{i^*j}} \hat{v}_{i^*j}.
\end{equation}

Now we combine the above three reduction steps together
and consider the event when mechanism $\cM$ cannot distinguish the $\cI^z$'s from each other.
When this happens,
mechanism $\cM$ produces the same outcome for all the instances.
Accordingly, although mechanism $\cM^*$ is given the distributions $\cD^z_{i^*}$,
by simulating $\cM$, it still produces the same outcome for all the $\cI^z_{i^*}$'s,
thus the same menu for all of them.
So mechanism $\cM^+$ also produces the same menu for all the $\cI^{+z}_{i^*}$'s:
that is, the menu produced by $\cM^*$ with the payments discounted by $1-\hat{\epsilon}$.
As a result, although  mechanism $\cM'$ is given the $\cD^z_{i^*j^*}$'s,
it still cannot ``distinguish'' the $\cI^z_{i^*j^*}$'s from each other and produces the
same outcome for all of them.
Using the similar argument in the proof of Lemma~\ref{lem:adaptive:single},
in this case
there exists $z^* \in [\lceil 4c \rceil]$ such that
$$Rev(\cM'(\cI^{z^*}_{i^*j^*})) < \frac{1}{2c} OPT(\cI^{z^*}_{i^*j^*}).$$
Combining this inequality with Equations \ref{eq:cm1}, \ref{eq:cm2} and \ref{eq:cm3}, we have
\begin{eqnarray}\label{eq:revz*}
&& Rev(\cM(\cI^{z^*})) \leq Rev(\cM^*(\cI^{z^*}_{i^*}))
+ \bE_{\hat{v} \sim \cD^{z^*}} \sum_{i\neq i^*}
v_i((\hat{v}_{ij})_{j\in [m]}, x_{i}(\hat{v})) \nonumber\\
&\leq& \frac{Rev(\cM^+(\cI^{+z^*}_{i^*}))}{1-\hat{\epsilon}} + \bar{\delta} / \hat{\epsilon}
+ \bE_{\hat{v} \sim \cD^{z^*}} \sum_{i\neq i^*}
v_i((\hat{v}_{ij})_{j\in [m]}, x_{i}(\hat{v})) \nonumber\\
&\leq& \frac{1}{1-\hat{\epsilon}}\left(Rev(\cM'(\cI^{z^*}_{i^*j^*}))
+ \sum_{j\neq j^*} \bE_{\hat{v}_{i^*j} \sim \cD^{z^*}_{i^*j}} \hat{v}_{i^*j}\right)
+ \bar{\delta} / \hat{\epsilon}
+ \bE_{\hat{v} \sim \cD^{z^*}} \sum_{i\neq i^*}
v_i((\hat{v}_{ij})_{j\in [m]}, x_{i}(\hat{v})) \nonumber\\
&<& \frac{1}{1-\hat{\epsilon}}\left(\frac{1}{2c}OPT(\cI^{z^*}_{i^*j^*})
+ \sum_{j\neq j^*} \bE_{\hat{v}_{i^*j} \sim \cD^{z^*}_{i^*j}} \hat{v}_{i^*j}\right)
+ \bar{\delta} / \hat{\epsilon}
+ \bE_{\hat{v} \sim \cD^{z^*}} \sum_{i\neq i^*}
v_i((\hat{v}_{ij})_{j\in [m]}, x_{i}(\hat{v})). 
\end{eqnarray}

Note that $OPT(\cI^{z^*}_{i^*j^*}) \leq OPT(\cI^{z^*})$, since selling a single item to a single player is a feasible outcome.
Moreover, since $\cD^{z^*}_{ij}$ is constantly $1$ when $i\neq i^*$ or $j\neq j^*$,
and since the
valuation function profile $v$ is succinct sub-additive, we have
\begin{eqnarray*}
& & \sum_{j\neq j^*} \bE_{\hat{v}_{i^*j} \sim \cD^{z^*}_{i^*j}} \hat{v}_{i^*j} = (m - 1)\ln H,\\
& & \bar{\delta} / \hat{\epsilon} =
\frac{1}{\hat{\epsilon}}
\bE_{\hat{v}_{i^*} \sim \cD^{z^*}_{i^*}}
\max_{S \subseteq [m]}
(\sum_{j \in S} \hat{v}_{i^* j} - v_{i^*}((\hat{v}_{i^*j})_{j\in [m]}, S))
\leq \frac{m \ln H}{\hat{\epsilon}},\\
& & \bE_{\hat{v} \sim \cD^{z^*}} \sum_{i\neq i^*}
v_i((\hat{v}_{ij})_{j\in [m]}, x_{i}(\hat{v}))
\leq m \ln H.
\end{eqnarray*}
Here the second equation is because $\sum_{j \in S} \hat{v}_{i^* j} - v_{i^*}((\hat{v}_{i^*j})_{j\in [m]}, S) \leq \sum_{j \in [m]} \hat{v}_{i^* j}$ for any $\hat{v}_{i^*}$ and $S$.
The third equation is because $\sum_{i\neq i^*}
v_i((\hat{v}_{ij})_{j\in [m]}, x_{i}(\hat{v})) \leq 
\sum_j \sum_{i\neq i^*} x_{ij}(\hat{v}) \hat{v}_{ij}$ for any $\hat{v}$: indeed, each item can be sold to at most one player, generating expected value $\ln H$.

%

\noindent
Combining the equations above with Equation \ref{eq:revz*}, we have
\begin{equation}
Rev(\cM(\cI^{z^*})) < \frac{1}{1-\hat{\epsilon}}(\frac{1}{2c}OPT(\cI^{z^*})
+ (m - 1)\ln H)
+ \frac{m \ln H}{\hat{\epsilon}}
+ m \ln H. \nonumber
\end{equation}

\noindent
Setting $\hat{\epsilon} = \frac{1}{4}$, we have
\begin{equation}\label{equ:37}
Rev(\cM(\cI^{z^*})) < \frac{2}{3c}OPT(\cI^{z^*}) + \frac{19m \ln H}{3}.
\end{equation}

Now we combine Equation \ref{equ:37} with the probability that
$\cM$ cannot distinguish the $\cI^z$'s.
As $OPT(\cI^{z^*})\geq 4cx$,
when $x > \frac{57 m \cdot c \cdot \ln H}{8}$, we have
$OPT(\cI^{z^*}) > \frac{57}{2}m \cdot c^2 \cdot \ln H$
and
$$Rev(\cM(\cI^{z^*}))\leq (1-\frac{1}{3c})(\frac{2}{3c}OPT(\cI^{z^*}) + \frac{19m \ln H}{3})
+ \frac{1}{3c}OPT(\cI^{z^*}) < \frac{1}{c}OPT(\cI^{z^*}).$$
Letting $C = \frac{1}{24c(4c+2)}$, 
which is a constant with respect to $c$, 
we have $Cnm \log_{m\ln H} H \leq C'nm\log_c H$.
Thus finishes the proof of Theorem \ref{thm:adaptive:multi}.
\end{proof}

When there are multiple items to sell, for non-adaptive queries, 
we can make some changes to the constructed distributions to improve the query complexity bound. 
Formally, consider the $\lceil 4c \rceil$ Bayesian instances $\{\cI^z = (N, M, \cD^z)\}_{z\in[\lceil 4c \rceil]}$ 
that are indistinguishable using only $Cnm\log_c H$ non-adaptive queries. 
Here for $i \neq i^{*}$ or $j \neq j^{*}$, 
$\cD^z_{ij}$ is the distribution that is constantly 1. 
For distribution $\cD^z_{i^*j^*}$, 
we construct them by shifting them in the interval with no value or quantile queries, i.e. $(u_s, u_{s+1})$ and $(q_s, q_{s+1})$,  
based on the equal revenue distribution with cumulative probability function
$\max\{0, 1 - \frac{\sqrt{H}}{v}\}$. 
Note that the optimal revenue for those distributions is at least $\sqrt{H}$. 
Applying the same analysis as in Theorem \ref{thm:adaptive:multi}, we have the following theorem, with proof omitted.

\begin{theorem}\label{thm:lower:multi}
For any constant $c > 1$,
there exists a constant $C$ such that,
for any $n\geq 1, m\geq 1$,
 any large enough $H$,
any succinct sub-additive valuation function profile $v = (v_i)_{i\in [n]}$,
and any DSIC Bayesian mechanism $\cM$
making less than $C nm \log_c H$
non-adaptive value and quantile queries to the oracle,
there exists a multi-item Bayesian auction instance $\cI = (N, M, \cD)$ with valuation profile $v$,
where $|N|=n, |M|=m$ and
the item values are bounded in $[1,H]$, such that
$Rev(\cM(\cI)) < \frac{OPT(\cI)}{c}$.
\end{theorem}

Succinct sub-additive valuations is a very broad class and contains single-item,
unit-demand, and additive auctions as special cases.
Thus Theorem \ref{thm:adaptive:multi} and \ref{thm:lower:multi} automatically apply to those cases.
We also note that it is shown in \cite{yao2017dominant}
that the optimal BIC revenue exceeds the optimal DSIC revenue by a constant factor even for two i.i.d. additive bidders and two identical items.
So even with infinite samples, there exist constants $c > 1$ such that no $c$-approximation to $OPT$ is possible.
However, Theorem \ref{thm:lower:multi} is stronger:
for {\em every} constant $c > 1$,
one needs at least the given number of queries to get a $c$-approximation.

\section{The Query Complexity for Bounded Distributions}
\label{sec:bounded}

In this section, we consider settings where all distributions are bounded within $[1, H]$,
and we construct efficient query mechanisms whose query complexity matches our lower-bounds.
We show that it is sufficient to use only value queries, and we define in Section \ref{subsec:vq}
a universal query scheme~$\cA_{V}$, which will
be used as a black-box in our mechanisms.
The seller uses algorithm $\cA_V$ to learn a distribution $\cD' = \times_{i\in N,j\in M} \cD'_{ij}$
that approximates the prior distribution $\cD$ and is stochastically dominated by $\cD$.
The seller then runs existing DSIC Bayesian mechanisms using $\cD'$, while the players' values are drawn from $\cD$.
In this sense, all our mechanisms are simple, but they are not given a true Bayesian instance as input.


The multi-player single-item setting is already non-trivial, but still easy,
 since we have a good understanding of the optimal mechanism, which is Myerson's auction~\cite{myerson1981optimal}.
In particular, in the analysis it suffices to apply the  {\em revenue monotonicity} theorem of \cite{devanur2016sample}.
The situation for unit-demand auctions and additive auctions is much more subtle.
The optimal auction could be very complicated
and may involve lotteries and bundling,
and revenue monotonicity may not hold \cite{hart2012maximal}.
Even (disregarding complexity issues and) assuming
we can design an optimal Bayesian mechanism for $\cD'$,
it is unclear how much revenue it guarantees when the players' values come from the true distribution $\cD$.
To overcome this difficulty, we rely on
 recent developments on {\em simple} Bayesian mechanisms with approximately optimal revenue.

The mechanism for unit-demand auctions is sequential post-price~\cite{kleinberg2012matroid} and the analysis is relatively easy.
For additive auctions, the Bayesian mechanism either runs Myerson's auction separately
for each item or runs the VCG mechanism with a per-player entry fee~\cite{yao2015n, cai2016duality}.
However, an easy and direct analysis would lose a factor of $m$ in the query complexity.
To achieve a tight upper-bound, we need to really open the box and
analyze the mechanism differently in several crucial places, exploring its behavior under oracle queries.

To sum up, given our query scheme, our mechanisms are black-box reductions to simple Bayesian mechanisms, thus are
simple, natural, and easy to implement in practice,
while the analysis is non-black-box, non-trivial and reveals interesting connections between Bayesian mechanisms and query schemes.

\subsection{The Value-Query Algorithm}
\label{subsec:vq}

The query algorithm~$\cA_{V}$ is defined in Algorithm \ref{QueryProcedure:value}.
Here $D\in \Delta(\bR)$ is the distribution to be queried. The
algorithm takes two parameters, the value bound $H$ and the precision factor $\delta>0$, makes $O(\log_{1+\delta}H)$
value queries to the oracle, and then returns a discrete distribution $D'$.
It is easy to verify that $D'$ is stochastically dominated by $D$.
Moreover, it is worth mentioning that
the mean of $D'$ approximates that of $D$.
Indeed, $mean(D)=\int_{1}^{H}vdF(v)\leq \sum_{l=0}^{k-1}v_{l+1}Pr[v_{l}\leq v <v_{l+1}] +v_{k}D(v_{k})
\leq (1+\delta)\sum_{l=0}^{k}v_{l}D'(v_{l})=(1+\delta)mean(D')$.
Therefore, by directly applying the parametric mechanism in \cite{azar2013optimal} with parameter $mean(D')$ (for single-parameter auctions where the distributions are regular or MHR),
we will get at least a $(1+\delta)$ fraction of their revenue.

\begin{algorithm}[htbp]
  \caption{\hspace{-3pt} The Value-Query Algorithm $\cA_{V}$}
 \label{QueryProcedure:value}
  \begin{algorithmic}[1]
\REQUIRE  The value bound $H$ and the precision factor $\delta$.

\STATE Let $k = \lceil \log_{1+\delta}H\rceil$ and define the {\em value vector} as
$v=(v_{0},v_{1},\dots, v_{k-1}, v_{k})=(1,(1+\delta), (1+\delta)^{2}, \dots, (1+\delta)^{k-1}, H).$\\

\STATE Query the oracle for $D$ with $v$, and receive a non-increasing quantile vector
$q = (q(v_{0}),\cdots,q(v_{k}))=(q_{l})_{l\in \{0, \dots, k\}}$. Note $q_0 = 1$.

\STATE Construct a discrete distribution $D'$ as follows:
$D'(v_{l}) = q_{l} - q_{l+1}$ for any $l\in \{0, \dots, k\}$, where $q_{k+1}\triangleq 0$.

\ENSURE Distribution $D'$.
\end{algorithmic}
\end{algorithm}

\subsection{Single-Item Auctions and Unit-Demand Auctions}

Denoting by $\cM_{MRS}$ Myerson's mechanism for single-item auctions,
Mechanism \ref{alg:single:value_query} defines our {\em efficient value Myerson} mechanism $\cM_{EVM}$.

\begin{algorithm}[htbp]
\floatname{algorithm}{Mechanism}
  \caption{\hspace{-3pt} Efficient Value Myerson Mechanism $\cM_{EVM}$}
 \label{alg:single:value_query}
   \begin{algorithmic}[1]
\STATE Given the value bound $H$ and a constant $\epsilon > 0$,
run the value-query algorithm $\cA_{V}$ with $H$ and $\delta = \epsilon$
for each player $i$'s distribution $\cD_{i}$.
Denote by $\cD'_{i}$ the returned distribution. Let $\cD'=\times_{i\in N}\cD'_{i}$.

\STATE Run $\cM_{MRS}$ with $\cD'$ and the players' reported values, $b=(b_{i})_{i\in N}$,
to get allocation $x = (x_{i})_{i\in N}$ and price profile $p = (p_i)_{i\in N}$ as the outcome.
  \end{algorithmic}
\end{algorithm}

The query complexity of $\cM_{EVM}$ is $O(n\log_{1+\epsilon} H)$,
since each distribution~$\cD_i$ needs $O(\log_{1+\epsilon} H)$ value queries in $\cA_{V}$.
When $\epsilon$ is sufficiently small, $O(n\log_{1+\epsilon} H)\approx O(n\epsilon^{-1}\log H)$.
Also, $\cM_{EVM}$ is DSIC since $\cM_{MRS}$ is so.

In this section and throughout the paper, we often analyze ``mismatching'' cases
where a Bayesian mechanism $\cM$ uses
distribution $\cD'$ while the actual Bayesian instance is $\cI = (N, M, \cD)$ (i.e.,
the players' true values are drawn from~$\cD$).
We use $Rev(\cM(\cI; \cD'))$ to denote the expected revenue in this case.
By construction, $Rev(\cM_{EVM}(\cI))= Rev(\cM_{MRS}(\cI; \cD'))$.

Because the distribution $\cD'$ constructed in $\cM_{EVM}$ is stochastically dominated by $\cD$,
letting $\cI' = (N, M, \cD')$ be the Bayesian instance under $\cD'$,
by revenue monotonicity \cite{devanur2016sample} we have
$Rev(\cM_{MRS}(\cI;\cD')) \geq Rev(\cM_{MRS}(\cI'))$.
By Lemma~5 of~\cite{devanur2016sample},
$Rev(\cM_{MRS}(\cI')) \geq \frac{OPT(\cI)}{1+\epsilon}$.
Thus we have the following simple result.

\begin{theorem}
\label{thm:single:value_query}
$\forall \epsilon > 0$, for any single-item instance $\cI = (N,M,\cD)$
with values in $[1,H]$,
mechanism $\cM_{EVM}$ is DSIC, has query complexity
$O(n\log_{1+\epsilon} H)$, and
$Rev(\cM_{EVM}(\cI)) \geq \frac{ OPT(\cI)}{1+\epsilon}$.
\end{theorem}

The construction for unit-demand auctions is similar,
except
the seller
uses as a blackbox the DSIC mechanism of \cite{kleinberg2012matroid}, denoted by $\cM_{UD}$:
see Mechanism~\ref{alg:unit:valuequery}.

\begin{algorithm}[htbp]
\floatname{algorithm}{Mechanism}
  \caption{\hspace{-3pt}  Mechanism ${\cM}_{EVUD}$ for Unit-Demand Auctions}
  \label{alg:unit:valuequery}
  \begin{algorithmic}[1]
\STATE Given $H$ and $\epsilon > 0$,
run the value-query algorithm $\cA_{V}$ with $H$ and
$\delta = \epsilon$
for each player~$i$'s distribution $\cD_{ij}$ for each item $j$.
Denote by $\cD'_{ij}$ the returned distribution. Let $\cD'_{i}=\times_{j\in M}\cD'_{ij}$ and $\cD'=\times_{i\in N}\cD'_{i}$.

\STATE Run $\cM_{UD}$ with $\cD'$ and the players' reported values, $b=(b_{ij})_{i\in N,j\in M}$,
to get allocation $x = (x_{ij})_{i\in N, j\in M}$ and price profile $p = (p_i)_{i\in N}$ as the outcome.
  \end{algorithmic}
\end{algorithm}

The main difficulty for unit-demand auctions is that we no longer have revenue monotonicity
as in single-item auctions.
Our analysis then comes in a non-blackbox way
and relies on the {\it COPIES setting} \cite{chawla2010multi,kleinberg2012matroid},
which provides an upper-bound for the optimal BIC revenue.
By properly upper-bounding the optimal revenue in the COPIES setting under $\cD'$,
we are able to upper-bound the optimal revenue in unit-demand auctions using the expected revenue of $\cM_{EVUD}$.
More precisely, we have the following theorem.

\begin{theorem}
\label{thm:unit-demand:value_query}
$\forall \epsilon > 0$, for any unit-demand instance $\cI = (N,M,\cD)$
with values in $[1,H]$,
mechanism $\cM_{EVUD}$ is DSIC, has
query complexity $O(mn\log_{1+\epsilon} H)$, and
$Rev(\cM_{EVUD}(\cI)) \geq \frac{OPT(\cI)}{24(1+\epsilon)}$.
\end{theorem}
Before analyzing mechanism $\cM_{EVUD}$, let us first recall the  {\em sequential post-price} mechanism $\cM_{UD}$.
This mechanism processes the players one by one according to
an arbitrary order, computes a price for each player $i$ based on remaining items, remaining players and the prior distribution,
and lets $i$ choose his utility-maximizing item (or choose none).
The revenue of this mechanism is analyzed by reducing the unit-demand instance to the {\em COPIES} setting, which we introduce below.

For a unit-demand auction instance $\cI=(N,M,\cD)$, the corresponding COPIES instance is denoted by $\cI^{CP}=(N^{CP}, M^{CP},\cD)$,
where each player $i\in N$ has $m$ copies and each item $j\in M$ has $n$ copies,
and player $i$'s copy~$j$
 is only interested in item $j$'s copy $i$, with value $v_{ij}$ drawn independently from~$\cD_{ij}$.
Thus $N^{CP} = M^{CP} = N\times M$,
and $\cI^{CP}$ is a single-parameter instance.
Denote by $N_{i}$ the set of player $i$'s copies and by $M_{j}$ the set of item $j$'s copies.
Note that both $\{N_{i}\}_{i\in N}$ and $\{M_{j}\}_{j\in M}$ are partitions of $N^{CP}$ (and $M^{CP}$).
Two natural constraints are imposed on feasible allocations under the COPIES setting, so as to connect it with the original unit-demand setting:
(1) for each player $i$, at most one of his copies gets an item; and
(2) for each item $j$, at most one of its copies gets allocated.
Accordingly, letting $q_{s}$ be the probability that a feasible mechanism allocates an item to a player copy $s\in N^{CP}$, we have
$\sum_{s\in N_{i}} q_{s} \leq 1$ for each $i\in N$ and $\sum_{s\in M_{j}} q_{s} \leq 1$ for each $j\in M$.

The corresponding mechanism $\cM_{UD}^{CP}$ for the COPIES setting works in the same way as $\cM_{UD}$, except that
it considers an arbitrary order of the players in $N^{CP}$, thus different copies of the same player may not be processed together.
When evaluating the performance of mechanism $\cM_{UD}^{CP}$,
the order of the players is chosen
by an {\em online adaptive adversary,}
who tries to minimize the expected revenue of the mechanism.
Because this adversary is the worst-case for mechanism $\cM_{UD}^{CP}$,
$$Rev(\cM_{UD}(\cI;\cD')) \geq Rev(\cM_{UD}^{CP}(\cI^{CP};\cD'))$$
for any distribution $\cD'$, where the latter is the expected revenue of $\cM_{UD}^{CP}$
under the online adaptive adversary.
Indeed,
mechanism $\cM_{UD}$ can be considered as $\cM_{UD}^{CP}$ under a specific order where all copies of each player come together,
thus the revenue is at least that when the order of $N^{CP}$ is adaptively chosen by the adversary.
Now we are ready to prove Theorem \ref{thm:unit-demand:value_query}.

\begin{proof}
[Proof of Theorem \ref{thm:unit-demand:value_query}]
It is easy to see that the query complexity of $\cM_{EVUD}$ is
$O(mn\log_{1+\epsilon} H)$,
since each distribution $\cD_{ij}$ needs $O(\log_{1+\epsilon} H)$ value queries. Also, it is immediate that $\cM_{EVUD}$ is DSIC.

Below we prove the revenue bound.
By construction,
\begin{equation}\label{eq:value:ud:cnstct}
Rev(\cM_{EVUD}(\cI)) = Rev(\cM_{UD}(\cI;\cD')).
\end{equation}
Let $\cI'=(N,M,\cD')$ and $\cI'^{CP}=(N^{CP},M^{CP},\cD')$.
We state the following key lemma, which is proved after the proof of Theorem \ref{thm:unit-demand:value_query}.

\begin{lemma} \label{lem:value:ud:chain}
$Rev(\cM_{UD}(\cI;\cD'))
\geq Rev(\cM_{UD}^{CP}(\cI^{CP};\cD'))
= Rev(\cM_{UD}^{CP}(\cI'^{CP})).$
\end{lemma}

By Theorem 1 of \cite{kleinberg2012matroid},
the sequential post-price mechanism is at least a 6-approximation to the optimal BIC revenue in the COPIES setting. Thus
\begin{equation}\label{eq:value:ud:thm1inK2012}
Rev(\cM_{UD}^{CP}(\cI'^{CP}))
\geq \frac{1}{6}OPT(\cI'^{CP}).
\end{equation}
Next, because the COPIES setting is a single-parameter setting, and because of the way we discretize the value space in algorithm $\cA_V$, by Lemma 5 of \cite{devanur2016sample} we have
\begin{equation}\label{eq:value:ud:useClaim}
OPT(\cI'^{CP})\geq \frac{1}{1+\epsilon} OPT(\cI^{CP}).
\end{equation}
Finally, by Theorem~6 of \cite{cai2016duality}, the optimal BIC revenue in the COPIES setting is a 4-approximation to the optimal BIC revenue
in the original unit-demand setting. Thus
\begin{equation}\label{eq:value:ud:useCPtoOrg}
OPT(\cI^{CP}) \geq \frac{1}{4}OPT(\cI).
\end{equation}

Combining Equations \ref{eq:value:ud:cnstct}, \ref{eq:value:ud:thm1inK2012}, \ref{eq:value:ud:useClaim}, \ref{eq:value:ud:useCPtoOrg}
and Lemma \ref{lem:value:ud:chain},
Theorem \ref{thm:unit-demand:value_query} holds.
\end{proof}

\begin{proof}[Proof of Lemma \ref{lem:value:ud:chain}]
The inequality is already explained.
Now we prove the equality.
For any value profile $v\sim\cD$, let
$v'$ be $v$ rounded down to the support of ${\cD'}$. That is,
for each~$v_{ij}$, $v'_{ij}$ is the largest value in the support of $\cD'_{ij}$ that is less than or equal to $v_{ij}$.
Recall that the support of $\cD'_{ij}$ is the set $\{v_{0}, \cdots, v_{k}\}$ as defined in the query algorithm $\cA_V$.
By
the definition of $\cD'_{ij}$, for any $0\leq l\leq k-1$,
$$\Pr_{v_{ij} \sim \cD_{ij}} [v'_{ij} =v_{l}] = \Pr_{v_{ij} \sim \cD_{ij}} [v_{ij} \geq v_{l}] - \Pr_{v_{ij} \sim \cD_{ij}} [v_{ij} \geq v_{l+1}]
=q(v_{l}) - q(v_{l+1}) = q_l-q_{l+1}=\cD'_{ij}(v_{l}),$$
 and
$$\Pr_{v_{ij} \sim \cD_{ij}} [v'_{ij} =v_{k}] = \Pr_{v_{ij} \sim \cD_{ij}} [v_{ij} \geq v_{k}] =q(v_{k}) = q_k = \cD'_{ij}(v_k).$$
That is, if $v$ is distributed according to $\cD$ then $v'$ is distributed according to ${\cD'}$.

For any value profile $v$ and the corresponding $v'$,
arbitrarily fix an order $\sigma$ of the players in $N^{CP}$,
which is a bijection from $\{1,\cdots,mn\}$ to $\{1,\cdots,mn\}$.
Without loss of generality, each player $\sigma(s)$ gets the corresponding item $\sigma(s)$ whenever
his true value is greater than or equal to the posted price for him.
Below we show that mechanism $\cM_{UD}^{CP}$ produces the same outcome no matter
the players' true values are $v$ or $v'$.
That is, for any $s\in \{1,\dots, mn\}$,
(1)
$\cM_{UD}^{CP}$ produces the same price $p_{\sigma(s)}$ under $v$ and $v'$ for player $\sigma(s)$,
and
(2) $v_{\sigma(s)}\geq p_{\sigma(s)}$ if and only if $v'_{\sigma(s)}\geq p_{\sigma(s)}$.

To prove these two properties, note that by the construction of mechanism $\cM_{UD}^{CP}$,
the price $p_{\sigma(s)}$ posted to $\sigma(s)$ depends only on the distribution~$\cD'$ and
the set $A_{\sigma(s)}$ of items sold to the players arriving before $\sigma(s)$.
Here $p_{\sigma(s)}$ may be randomized if $\cD'_{\sigma(s)}$ is irregular,
but it always takes value in the support of $\cD'_{\sigma(s)}$ (except that,
if selling the corresponding item $\sigma(s)$ to player $\sigma(s)$ is not feasible anymore, then $p_{\sigma(s)} = + \infty$).

We prove the two desired properties by induction.
When $s=1$, property (1) trivially holds, because $A_{\sigma(1)} = \emptyset$ under both value profiles.
Furthermore,
because a realization of $p_{\sigma(1)}$ is always in
 the support of $\cD'_{\sigma(1)}$,
 and because $v'_{\sigma(1)}$ is $v_{\sigma(1)}$ rounded down to the support of $\cD'_{\sigma(1)}$,
property (2) holds
when $s=1$.

Now assume (1) and (2) hold for any $s\leq t$ with $t < mn$.
We show they also hold for $s=t+1$.
Indeed, the inductive hypothesis implies that
for any $s\leq t$, $A_{\sigma(s)}$ is the same under the two value profiles.
In particular, $A_{\sigma(t+1)}$ is the same,
which means the price $p_{\sigma(t+1)}$ is the same. Thus property (1) holds.
Property (2) also holds because a realization of $p_{\sigma(t+1)}$ is always in the support of $\cD'_{\sigma(t+1)}$.
In sum, for any order $\sigma$, mechanism $\cM_{UD}^{CP}$ produces the same outcome under the two value profiles $v$ and $v'$,
thus the same revenue.

Accordingly, under the online adaptive adversary for $(\cI^{CP}; \cD')$,
the revenue $Rev(\cM_{UD}^{CP}(\cI^{CP};\cD'))$
is the same as the revenue
when the players' true values are obtained by rounding $v\sim \cD$ to $v'$.
Because the resulting $v'$ is distributed according to $\cD'$,
$Rev(\cM_{UD}^{CP}(\cI^{CP};\cD'))$ is at least the expected revenue of $\cM_{UD}^{CP}$ under the online adaptive adversary for
$\cI'^{CP}$.
Indeed, a randomized adversary for $\cI'^{CP}$ can simulate the
adversary for $(\cI^{CP}; \cD')$: in each step, given $v'_s$ with $s\in N^{CP}$ being the player in this step,
the former first samples $v_s$ from $\cD_s$ conditional on $v_s$ rounded down to $v'_s$,
and then uses the latter to decide which player arrives next.
Thus,
$$Rev(\cM_{UD}^{CP}(\cI^{CP};\cD')) \geq Rev(\cM_{UD}^{CP}({\cI'}^{CP})).$$
Similarly,
$$Rev(\cM_{UD}^{CP}(\cI^{CP};\cD')) \leq Rev(\cM_{UD}^{CP}({\cI'}^{CP})).$$
Therefore
$Rev(\cM_{UD}^{CP}(\cI^{CP};\cD'))  = Rev(\cM_{UD}^{CP}({\cI'}^{CP}))$ and Lemma  \ref{lem:value:ud:chain} holds.
\end{proof}

Letting $c=24(1+\epsilon)$, we have the query complexity in Table \ref{table:main results}.

\subsection{Additive Auctions}\label{subsec:additive}

For additive auctions,
the DSIC Bayesian mechanism in \cite{yao2015n, cai2016duality}
chooses between two mechanisms, whichever generates higher expected revenue  under the true prior $\cD$.
The first is the ``individual Myerson'' mechanism, denoted by $\cM_{IM}$, which
sells each item separately using Myerson's mechanism.
The second is
the VCG mechanism with optimal per-player entry fees, denoted by~$\cM_{BVCG}$.

In our mechanism $\cM_{EVA}$,
the seller queries about $\cD$ using algorithm~$\cA_V$ with properly chosen parameters.
Given the resulting distribution $\cD'$,
the seller either runs $\cM_{IM}$ or runs $\cM_{BVCG}$ as a blackbox,
resulting in query mechanisms $\cM_{EVIM}$ and $\cM_{EVBVCG}$.
 We only define the latter in Mechanism~\ref{mechanism:evbvcg},
 and the former simply replaces $\cM_{BVCG}$ with $\cM_{IM}$.
Note that $Rev(\cM_{EVIM}(\cI)) = Rev(\cM_{IM}(\cI; \cD'))$ and
 $Rev(\cM_{EVBVCG}(\cI)) = Rev(\cM_{BVCG}(\cI; \cD'))$.
However, the seller cannot compute these two revenue and choose the better one, because
he does not know $\cD$.
Thus he randomly chooses between the two, according to
probabilities defined in our analysis, to optimize the approximation ratio. 
We have the following theorem. 

Theorem \ref{thm:veca} is harder to show. Indeed, one cannot use revenue monotonicity or the 
COPIES setting to easily upper-bound the optimal BIC revenue.
Our analysis is based on
the duality framework of \cite{cai2016duality} for Bayesian auctions, properly adapted for the query setting.
Finally, letting $c = 8(1+\epsilon)$, we have the query complexity in Table \ref{table:main results}.

\begin{algorithm}
\floatname{algorithm}{Mechanism}
  \caption{\hspace{-3pt} Mechanism ${\cM}_{EVBVCG}$ to Approximate $\cM_{BVCG}$ via Value Queries}
  \label{mechanism:evbvcg}
    \begin{algorithmic}[1]
\STATE Given $H$ and $\epsilon > 0$,
run the value-query algorithm $\cA_{V}$ with $H$ and $\delta = \sqrt{\epsilon+1}-1$
for each player~$i$'s distribution $\cD_{ij}$ for each item $j$.
Denote by $\cD'_{ij}$ the returned distribution. Let $\cD'_{i}=\times_{j\in M}\cD'_{ij}$ and $\cD'=\times_{i\in N}\cD'_{i}$.

\STATE Run $\cM_{BVCG}$ with $\cD'$ and the players' reported values, $b=(b_{ij})_{i\in N, j\in M}$,
to get allocation $x = (x_{ij})_{i\in N, j\in M}$ and price profile $p = (p_i)_{i\in N}$ as the outcome.
  \end{algorithmic}
\end{algorithm}


\begin{theorem}\label{thm:veca}
$\forall \epsilon > 0$, for any additive instance $\cI = (N,M,\cD)$
with values in $[1,H]$,
mechanism ${\cM}_{EVA}$ is DSIC, has query complexity $O(mn\log_{1+\epsilon} H)$, and
$Rev({\cM}_{EVA}(\cI))
\geq \frac{OPT(\cI)}{8(1+\epsilon)}$.
\end{theorem}

\begin{proof}
First, it is easy to see that the query complexity of mechanism ${\cM}_{EVA}$ is $O(mn\log_{1+\delta} H)$, since there are in total $mn$ distributions
and each one of them needs $O(\log_{1+\delta} H)$ value queries in the algorithm $\cA_{V}$.
Since $\delta=\sqrt{\epsilon+1}-1$, $O(mn\log_{1+\delta} H) = O(mn\log_{1+\epsilon} H)$.
Second, since mechanisms $\cM_{BVCG}$ and $\cM_{IM}$ are both DSIC, ${\cM}_{EVA}$ is DSIC.

Recall that
mechanism ${\cM}_{EVA}$ randomly chooses between running $\cM_{EVIM}$ and running $\cM_{EVBVCG}$.
Therefore, to upper-bound the optimal revenue $OPT(\cI)$ using $Rev({\cM}_{EVA}(\cI))$,
we only need to upper-bound $OPT(\cI)$
using $Rev(\cM_{EVIM}(\cI))$ and $Rev(\cM_{EVBVCG}(\cI))$.

As in \cite{cai2016duality}, we only need to consider the prior distribution $\cD$ with finite support.
Let $V_{ij}$ be the support of $\cD_{ij}$ for each player $i$ and item $j$,
$V_{i}=\times_{j\in M} V_{ij}$ and $V=\times_{i\in N} V_i$.
In the optimal BIC mechanism,
when player $i$ bids $v_i$,
let $\pi_{ij}(v_i)$ be the probability for him to get item $j$
and $p_{i}(v_{i})$ be his expected payment,
taken
over the randomness of the other players' values and the randomness of the mechanism.
Let $\pi=(\pi_{ij}(v_{i}))_{i\in N, j\in M, v_{i}\in V_{i}}$ and $p=(p_{i}(v_{i}))_{i\in N, v_{i}\in V_{i}}$.
The pair $(\pi,p)$ is called the {\em reduced form} (of the optimal BIC mechanism) \cite{cai2012algorithmic}.

Denote by $\tilde{\varphi}_{ij}(v_{ij})$ Myerson's (ironed) virtual value when player $i$'s value on item $j$ is $v_{ij}$.
For any value sub-profile $v_{-i}$ of the players other than $i$,
let $\beta_{ij}(v_{-i}) = \max_{i'\neq i} v_{i'j}$: that is, the highest bid on item $j$ excluding player $i$.
Moreover, let $r_{ij} (v_{-i}) = \max_{x \geq \beta_{ij} (v_{-i})} \{x \cdot \Pr_{v_{ij} \sim \cD_{ij}} [v_{ij} \geq x]\}$,
$r_i(v_{-i}) = \sum_j r_{ij}(v_{-i})$, $r_i = \bE_{v_{-i} \sim \cD_{-i}} [r_i(v_{-i})]$, and finally $r = \sum_i r_i$.
Note that $r$ is the expected revenue by running the 1-look-ahead mechanism of \cite{ronen2001approximating} for each item separately,
and $r \leq Rev(\cM_{IM}(\cI))$. 

Next, we use a different method from \cite{cai2016duality} to partition each player $i$'s value space $V_{i}$
into $m+1$ subsets.
More precisely, given $\delta>0$ and $v_{-i}$, let $R^{(v_{-i})}_{0}=\{  v_{i}\in V_i \  | \ v_{ij} < (1+\delta) \beta_{ij}(v_{-i}),\forall j\}$.
For any $v_i\notin R^{(v_{-i})}_{0}$, let $j =  \argmax\{v_{ij} - (1+\delta) \beta_{ij}(v_{-i})\}$ with ties broken lexicographically,
and add $v_i$ to the set $R^{(v_{-i})}_{j}$: note that $v_{ij} - (1+\delta) \beta_{ij}(v_{-i}) \geq 0$ in this case.
Similar to Theorem 3 of~\cite{cai2016duality}, the optimal BIC revenue can be upper-bounded by the sum of the following terms,
where $\cD_i(v_i)$ and $\cD_{-i}(v_{-i})$ are respectively the probabilities of $v_i$ and $v_{-i}$ under $\cD$, and ${\bf I}$
is the indicator function:
\begin{equation}\label{eq:add:val:decompose}
OPT(\cI) \leq \mbox{Single $+$ Under $+$ Over $+$ Tail $+$ Core},
\end{equation}
where
\begin{equation*}
\mbox{Single} =
\sum_i \sum_{v_i\in V_i} \sum_j \cD_i(v_i) \cdot \pi_{ij}(v_i)
\cdot \tilde{\varphi}_{ij}(v_{ij}) \cdot \Pr_{v_{-i}\sim \cD_{-i}}[v_i\in R_j^{(v_{-i})}],
\end{equation*}

\begin{equation*}
\mbox{Under} =
 \sum_i \sum_{v_i\in V_i} \sum_j \cD_i(v_i) \cdot \pi_{ij}(v_i)
\cdot \sum_{v_{-i}\in V_{-i}} v_{ij} \cdot \cD_{-i}(v_{-i})
{\bf I}_{v_{ij}<(1+\delta)\beta_{ij}(v_{-i})},
\end{equation*}

\begin{equation*}
\mbox{Over} =
\sum_i \sum_{v_i\in V_i} \sum_j \cD_i(v_i) \cdot \pi_{ij}(v_i) \cdot
\sum_{v_{-i}\in V_{-i}} (1+\delta)\beta_{ij}(v_{-i})\cD_{-i}(v_{-i})
{\bf I}_{v_{ij} \geq (1+\delta)\beta_{ij}(v_{-i})},
\end{equation*}

\begin{equation*}
\begin{aligned}
\mbox{Tail} =
\sum_i \sum_{v_{-i}\in V_{-i}} \cD_{-i}(v_{-i}) \sum_j
\sum_{v_{ij} > (1+\delta)\beta_{ij}(v_{-i}) + r_i(v_{-i})} \cD_{ij}(v_{ij}) \cdot (v_{ij}-(1+\delta)\beta_{ij}(v_{-i})) \\
\quad\quad\quad \cdot \Pr_{v_{i,-j}\sim \cD_{i,-j}}
[\exists k\neq j, v_{ik} - (1+\delta)\beta_{ik}(v_{-i}) \geq v_{ij} - (1+\delta)\beta_{ij}(v_{-i})],
\end{aligned}
\end{equation*}
and
\begin{equation*}
\begin{aligned}
\mbox{Core} =
\sum_i \sum_{v_{-i}\in V_{-i}} \cD_{-i}(v_{-i}) \sum_j
\sum_{(1+\delta)\beta_{ij}(v_{-i})\leq v_{ij}\leq (1+\delta)\beta_{ij}(v_{-i}) + r_i(v_{-i})} \cD_{ij}(v_{ij})\\
\mbox{ $\quad\quad\quad \quad\quad\quad\quad \quad\quad\quad \quad\quad\quad\quad \quad\quad $}
\cdot (v_{ij}-(1+\delta)\beta_{ij}(v_{-i})).
\end{aligned}
\end{equation*}

In the following, 
we bound these terms in Inequality \ref{eq:add:val:decompose} separately.
Note that when $\cM_{EVIM}$ uses the value-query algorithm $\cA_{V}$ to learn a distribution,
the parameters are also set to be $H$ and $\delta=\sqrt{\epsilon+1}-1$.
Thus, applying Theorem~\ref{thm:single:value_query} to each item, we have
$$ Rev(\cM_{IM}(\cI)) \leq (1+\delta)Rev(\cM_{EVIM}(\cI)).$$
For the terms Single, Under, Over and Tail, we are able to upper-bound them using $Rev(\cM_{EVIM}(\cI))$.
Following Lemma 13 of \cite{cai2016duality},
although the term Single has changed from its original form, we still have
\begin{eqnarray}
\mbox{Single} &=& \sum_i \sum_{v_i\in V_i} \sum_j \cD_i(v_i) \cdot \pi_{ij}(v_i)
\cdot \tilde{\varphi}_{ij}(v_{ij}) \cdot \Pr_{v_{-i}\sim \cD_{-i}}[v_i\in R_j^{(v_{-i})}] \nonumber\\
&\leq&Rev(\cM_{IM}(\cI)) \leq (1+\delta) Rev(\cM_{EVIM}(\cI)).\label{eq:add:val:single}
\end{eqnarray}

Next, using Lemmas 14 and 15 of \cite{cai2016duality}, we upper-bound the term Under as follows:
\begin{eqnarray}
\mbox{Under} &=& \sum_i \sum_{v_i\in V_i} \sum_j \cD_i(v_i) \cdot \pi_{ij}(v_i)
\cdot \sum_{v_{-i}\in V_{-i}} \cD_{-i}(v_{-i})\cdot v_{ij} \cdot
{\bf I}_{v_{ij}<(1+\delta)\beta_{ij}(v_{-i})} \nonumber\\
&=& \sum_i \sum_{v_i\in V_i} \sum_j \cD_i(v_i) \cdot \pi_{ij}(v_i)\cdot \sum_{v_{-i}\in V_{-i}} \cD_{-i}(v_{-i})\cdot v_{ij} \nonumber \\
& & \quad\quad\quad \quad\quad\quad\quad\quad\quad \quad\quad\quad \quad\quad\quad   \cdot ( {\bf I}_{v_{ij}<\beta_{ij}(v_{-i})} +
{\bf I}_{\beta_{ij}(v_{-i})\leq v_{ij}<(1+\delta)\beta_{ij}(v_{-i})}) \nonumber\\
&\leq& \sum_i \sum_{v_i\in V_i} \sum_j \cD_i(v_i) \cdot \pi_{ij}(v_i) \cdot \sum_{v_{-i}\in V_{-i}} \cD_{-i}(v_{-i})\nonumber\\
& & \quad\quad\quad \quad\quad\quad\quad\quad\quad \quad\quad\quad\quad  \cdot ( v_{ij} \cdot {\bf I}_{v_{ij}<\beta_{ij}(v_{-i})}+(1+\delta)\beta_{ij}(v_{-i}) \cdot {\bf I}_{v_{ij}\geq \beta_{ij}(v_{-i})} ) \nonumber\\
&\leq& Rev(\cM_{IM}(\cI))  + (1+\delta) Rev(\cM_{IM}(\cI)) \leq  2(1+\delta)^{2} Rev(\cM_{EVIM}(\cI)). \nonumber
\end{eqnarray}
The second inequality above is by Lemmas 14 and 15 of \cite{cai2016duality}, which respectively upper-bound
the term Over and the term Under in the original setting.
Indeed, we split our term Under into the sum of the original terms Under and Over.
Using the above equation, the approximation ratio to $OPT(\cI)$
will be $9(1+\epsilon)$ eventually.
To get the desired $8(1+\epsilon)$-approximation,
we prove a variant of Lemma 15 of \cite{cai2016duality}, which directly upper-bounds our term Under
as
\begin{equation}\label{eq:add:val:under}
\mbox{Under}  \leq (1+\delta) Rev(\cM_{IM}(\cI)) \leq (1+\delta)^{2} Rev(\cM_{EVIM}(\cI)).
\end{equation}
The actual proof of this alternative lemma is tedious and does not provide new insights to our result, thus has been omitted.

Next, we upper-bound the term Over:
\begin{eqnarray}
\mbox{Over} &=& \sum_i \sum_{v_i\in V_i} \sum_j \cD_i(v_i) \cdot \pi_{ij}(v_i)
\cdot \sum_{v_{-i}\in V_{-i}} (1+\delta)\beta_{ij} (v_{-i})\cD_{-i}(v_{-i})
{\bf I}_{v_{ij} \geq (1+\delta)\beta_{ij}(v_{-i})} \nonumber\\
&\leq& (1+\delta) \sum_i \sum_{v_i\in V_i} \sum_j \cD_i(v_i) \cdot \pi_{ij}(v_i)
\cdot \sum_{v_{-i}\in V_{-i}} \beta_{ij} (v_{-i})\cD_{-i}(v_{-i})
{\bf I}_{v_{ij} \geq \beta_{ij}(v_{-i})} \nonumber\\
&\leq& (1+\delta)Rev(\cM_{IM}(\cI)) \leq  (1+\delta)^{2} Rev(\cM_{EVIM}(\cI)). \label{eq:add:val:over}
\end{eqnarray}
The second inequality above is by Lemma 14 of \cite{cai2016duality}.

Next, we upper-bound the term Tail, which is similar to the analysis of \cite{cai2016duality}, but with
the threshold price $\beta_{ij}(v_{-i})$ scaled up by a factor of $(1+\delta)$.
\begin{eqnarray}
\mbox{Tail} &=& \sum_i \sum_{v_{-i}\in V_{-i}} \cD_{-i}(v_{-i}) \sum_j
\sum_{v_{ij} > (1+\delta)\beta_{ij}(v_{-i}) + r_i(v_{-i})} \cD_{ij}(v_{ij}) \cdot
(v_{ij}-(1+\delta)\beta_{ij}(v_{-i})) \nonumber\\
&&  \quad\quad \quad\cdot \Pr_{v_{i,-j}\sim \cD_{i,-j}}
[\exists k\neq j, v_{ik} - (1+\delta)\beta_{ik}(v_{-i}) \geq v_{ij} - (1+\delta)\beta_{ij}(v_{-i})] \nonumber\\
&\leq& \sum_i \sum_{v_{-i}\in V_{-i}} \cD_{-i}(v_{-i}) \sum_j
\sum_{v_{ij} > (1+\delta)\beta_{ij}(v_{-i}) + r_i(v_{-i})} \cD_{ij}(v_{ij}) \cdot
(v_{ij}-(1+\delta)\beta_{ij}(v_{-i})) \nonumber\\
&& \quad\quad\quad \quad\quad \quad\cdot \Pr_{v_{i,-j}\sim \cD_{i,-j}}
[\exists k\neq j, v_{ik} - \beta_{ik}(v_{-i}) \geq v_{ij} - (1+\delta)\beta_{ij}(v_{-i})] \nonumber\\
&\leq&  \sum_i \sum_{v_{-i}\in V_{-i}} \cD_{-i}(v_{-i}) \sum_j
\sum_{v_{ij} > (1+\delta)\beta_{ij}(v_{-i}) + r_i(v_{-i})} \cD_{ij}(v_{ij}) \cdot
(v_{ij}-(1+\delta)\beta_{ij}(v_{-i})) \nonumber\\
&&\quad \quad\quad \quad \quad\quad \quad\quad \quad \cdot \sum_{k=1}^{m}  \Pr_{v_{ik}\sim \cD_{ik}}  [ v_{ik}  \geq v_{ij} - (1+\delta)\beta_{ij}(v_{-i}) + \beta_{ik}(v_{-i})] \nonumber\\
 &\leq&  \sum_i \sum_{v_{-i}\in V_{-i}} \cD_{-i}(v_{-i}) \sum_j
\sum_{v_{ij} > (1+\delta)\beta_{ij}(v_{-i}) + r_i(v_{-i})}  \cD_{ij}(v_{ij}) \nonumber\\
&&\cdot  \sum_{k=1}^{m} (v_{ij}-(1+\delta)\beta_{ij}(v_{-i}) + \beta_{ik}(v_{-i}))
%
 \Pr_{v_{ik}\sim \cD_{ik}}  [ v_{ik}  \geq v_{ij} - (1+\delta)\beta_{ij}(v_{-i}) + \beta_{ik}(v_{-i})] \nonumber\\
&\leq& \sum_i \sum_{v_{-i}\in V_{-i}} \cD_{-i}(v_{-i}) \sum_j
\sum_{v_{ij} > (1+\delta)\beta_{ij}(v_{-i}) + r_i(v_{-i})} \cD_{ij}(v_{ij}) \sum_{k=1}^{m}  r_{ik}(v_{-i})  \nonumber\\
&=& \sum_i \sum_{v_{-i}\in V_{-i}} \cD_{-i}(v_{-i}) \sum_j r_{i}(v_{-i})
\sum_{v_{ij} > (1+\delta)\beta_{ij}(v_{-i}) + r_i(v_{-i})} \cD_{ij}(v_{ij})  \nonumber\\
&\leq& \sum_i \sum_{v_{-i}\in V_{-i}} \cD_{-i}(v_{-i}) \sum_j ((1+\delta)\beta_{ij}(v_{-i}) + r_{i}(v_{-i}))  \nonumber\\
&&\quad \quad\quad \quad  \quad\quad \quad\quad \quad\quad \quad \quad\quad \quad\quad \quad
\cdot \Pr_{v_{ij} \sim \cD_{ij}} [v_{ij} > (1+\delta)\beta_{ij}(v_{-i}) + r_i(v_{-i})] \nonumber\\
&\leq& \sum_i \sum_{v_{-i}\in V_{-i}} \cD_{-i}(v_{-i}) \sum_j  r_{ij}(v_{-i})  =
\sum_i \sum_{v_{-i}\in V_{-i}} \cD_{-i}(v_{-i}) r_i(v_{-i}) = \sum_i r_i \nonumber\\
&=& r \leq Rev(\cM_{IM}(\cI)) \leq  (1+\delta)Rev(\cM_{EVIM}(\cI)). \label{eq:add:val:tail}
\end{eqnarray}
The second inequality above is by union bound. The fourth and sixth inequalities use twice the definition of
$r_{ij}(v_{-i})$, which sets the optimal price to maximize the expected revenue generated by selling item~$j$ to $i$.
The second equality is by the definition of $r_{i}(v_{-i})$.

Finally, we upper-bound the term Core.
The Core part is the most complicated, and we use $\cM_{EVBVCG}$ and $\cM_{EVIM}$ together to upper-bound it.
To do so, below we rewrite Core into a different form.
Similar to~\cite{cai2016duality}, arbitrarily fixing $v_{-i}$ and letting $v_{ij}\sim \cD_{ij}$,
define the following two new random variables, which again scale the threshold price $\beta_{ij}(v_{-i})$ up by a factor of $(1+\delta)$:
$$b_{ij}(v_{-i}) = (v_{ij} -(1+\delta) \beta_{ij}(v_{-i}))
{\bf I}_{v_{ij}\geq (1+\delta)\beta_{ij}(v_{-i})},$$
and
$$c_{ij}(v_{-i}) = b_{ij}(v_{-i})
{\bf I}_{b_{ij}(v_{-i}) \leq r_i(v_{-i})}.$$
Therefore, we have
$$\mbox{Core} = \sum_i \sum_{v_{-i}\in V_{-i}} \cD_{-i}(v_{-i}) \sum_j \bE_{v_{ij} \sim \cD_{ij}}[c_{ij}(v_{-i})].$$
Letting $e_i(v_{-i}) = \sum_{j} \bE_{v_{ij} \sim \cD_{ij}}[c_{ij}(v_{-i})] - 2r_i(v_{-i})$,
following the proof of Lemma 12 in \cite{cai2016duality},
 we still have
$$\Pr[\sum_j b_{ij}(v_{-i}) \geq e_i(v_{-i})] \geq \frac{1}{2}.$$

In the following, we use the revenue of mechanisms ${\cM}_{EVBVCG}$ and $\cM_{EVIM}$ to bound the Core.
To do so, first note that by the construction of mechanism ${\cM}_{EVBVCG}$,
$$Rev({\cM}_{EVBVCG}(\cI)) = Rev(\cM_{BVCG}(\cI; \cD')).$$
Let $V'_{ij}$ be the support of $\cD'_{ij}$,  $V'_{i}=\times_{j\in M}V'_{ij}$, $V'=\times_{i\in N}V'_{i}$.
As before, given $v_i\sim \cD_i$, denote
 by $v'_{i}\in V'_{i}$ the value vector obtained by rounding $v_i$ down to the support of $\cD'_i$.
  That is, each $v'_{ij}$ is the largest value in $V'_{ij}$ that is less than or equal to $v_{ij}$.
Then,
\begin{eqnarray*}
Rev(\cM_{BVCG}(\cI; \cD'))& \geq& \sum_{i}\bE_{v_{-i}\sim\cD_{-i}} \bE_{v_{i}\sim\cD_{i}}Rev(\cM_{BVCG}(v_{i}',v_{-i}; \cD'))\\
&=&\sum_{i}\bE_{v_{-i}\sim\cD_{-i}} \bE_{v'_{i}\sim \cD'_{i}}Rev(\cM_{BVCG}(v_{i}',v_{-i}; \cD')).
\end{eqnarray*}
The inequality is because each player $i$ can potentially buy item $j$ only when $j$ is in his winning set
(i.e., he is the highest bidder for $j$),
and $i$'s winning set under $v'_{i}$ is a subset of his wining set under $v_{i}$.
Moreover, the entry fee of $i$ is the same under both $(v_i, v_{-i})$ and $(v'_i, v_{-i})$,
as it only depends on $\cD'_i$ and $v_{-i}$.
Thus the revenue inside the expectation does not increase when $v_i$ is replaced by~$v_i'$.
The equality is again because drawing $v_{i}$ from $\cD_{i}$ and then rounding down to $v'_{i}$
is equivalent to drawing $v'_{i}$ from $\cD'_{i}$ directly.

Next, we lower-bound $\sum_{i}\bE_{v_{-i}\sim\cD_{-i}} \bE_{v'_{i}\sim \cD'_{i}}Rev(\cM_{BVCG}(v_{i}',v_{-i}; \cD'))$.
As before, arbitrarily fixing $v_{-i}$ and letting $v'_{ij}\sim \cD'_{ij}$, define
$$b'_{ij}(v_{-i}) = (v'_{ij} - \beta_{ij}(v_{-i})) {\bf I}_{v'_{ij}\geq \beta_{ij}(v_{-i})}.$$
Note that
 $b'_{ij}(v_{-i})$ is a random variable that represents player $i$'s utility in the second price mechanism
on item $j$ with value $v'_{ij}\sim \cD'_{ij}$, when the other players' bids are $v_{-i,j}$.
Also note that $\cM_{BVCG}$ uses the optimal entry fee for each $i$ with respect to $v_{-i}$ and $\cD'$,
which generates expected revenue from $i$ (over $\cD'_i$) greater than or equal to that by using the following entry fee,
$$e'_i(v_{-i}) = \frac{e_i(v_{-i})}{1+\delta}.$$
Now we show player $i$ accepts the entry fee $e'_i(v_{-i}) $ with probability at least $\frac{1}{2}$. Indeed,
for any $v_i$ and the corresponding $v'_i$,
\begin{eqnarray*}
\sum_{j} b'_{ij}(v_{-i})& =& \sum_{j} (v'_{ij} - \beta_{ij}(v_{-i}))
{\bf I}_{v'_{ij}\geq \beta_{ij}(v_{-i})}
\geq \sum_{j} (\frac{v_{ij}}{1+\delta} - \beta_{ij}(v_{-i}))
{\bf I}_{\frac{v_{ij}}{1+\delta} \geq \beta_{ij}(v_{-i})}\\
&=& \frac{1}{{1+\delta} }\sum_{j} (v_{ij}- (1+\delta)\beta_{ij}(v_{-i}))
{\bf I}_{v_{ij} \geq (1+\delta)\beta_{ij}(v_{-i})}
= \frac{1}{1+\delta}\sum_{j}b_{ij}(v_{-i}).
\end{eqnarray*}
The inequality is because $v'_{ij} \geq \frac{v_{ij}}{1+\delta}$,
and because $\frac{v_{ij}}{1+\delta} \geq \beta_{ij}(v_{-i})$ implies $v'_{ij} \geq \beta_{ij}(v_{-i})$.
Therefore
\begin{eqnarray*}
\Pr_{v'_i\sim \cD'_i}[\sum_j b'_{ij}(v_{-i}) \geq e'_i(v_{-i})]
&\geq & \Pr_{v_i\sim \cD_i}[ \frac{1}{1+\delta} \sum_j b_{ij}(v_{-i}) \geq \frac{e_i(v_{-i})}{1+\delta} ]\\
&=&\Pr_{v_i\sim \cD_i}[\sum_j b_{ij}(v_{-i}) \geq e_i(v_{-i})] \geq \frac{1}{2},
\end{eqnarray*}
as desired.
Thus we have
\begin{eqnarray*}
&&Rev(\cM_{EVBVCG}(\cI))\\
& \geq &\sum_{i}\bE_{v_{-i}\sim\cD_{-i}} \bE_{v'_{i}\sim \cD'_{i}}Rev(\cM_{BVCG}(v_{i}',v_{-i}; \cD'))\\
& \geq & \frac{1}{2} \sum_i \sum_{v_{-i}\in V_{-i}} \cD_{-i}(v_{-i}) \cdot \frac{e_i(v_{-i})}{1+\delta} \\
&=& \frac{1}{2(1+\delta)} \sum_i \sum_{v_{-i}\in V_{-i}} \cD_{-i}(v_{-i}) \left( \sum_{j} \bE_{v_{ij} \sim \cD_{ij}}[c_{ij}(v_{-i})] - 2r_i(v_{-i}) \right)\\
&=& \frac{1}{2(1+\delta)}\mbox{Core} -\frac{r}{1+\delta}.
\end{eqnarray*}
That is,
\begin{eqnarray}
\mbox{Core} &\leq& 2(1+\delta) Rev(\cM_{EVBVCG}(\cI))+ 2r \nonumber \\
&\leq& 2(1+\delta) \left[ Rev(\cM_{EVBVCG}(\cI))+ Rev(\cM_{EVIM}(\cI))\right]. \label{eq:add:val:core}
\end{eqnarray}

Combining Inequalities \ref{eq:add:val:decompose}, \ref{eq:add:val:single}, \ref{eq:add:val:under}, \ref{eq:add:val:over},
\ref{eq:add:val:tail} and \ref{eq:add:val:core},
\begin{eqnarray*}
OPT(\cI) &\leq& (1+\delta)^{2}\left( 2 Rev(\cM_{EVBVCG}(\cI))+ 6 Rev(\cM_{EVIM}(\cI))\right)\\
&=& (1+\epsilon)\left( 2 Rev(\cM_{EVBVCG}(\cI))+ 6 Rev(\cM_{EVIM}(\cI))\right).
\end{eqnarray*}
Accordingly, by running mechanism $\cM_{EVBVCG}$ with probability $\frac{1}{4}$
and mechanism $\cM_{EVIM}$ with probability $\frac{3}{4}$,
the expected revenue of mechanism $\cM_{EVA}$ is
$$Rev({\cM}_{EVA}(\cI))
\geq \frac{1}{8(1+\epsilon)} OPT(\cI).
$$
This finishes the proof of Theorem \ref{thm:veca}.
\end{proof}

\section{The Query Complexity for Unbounded Distributions}\label{sec:queryunbounded}

Next, we construct efficient query mechanisms for arbitrary distributions whose supports can be unbounded.
For a mechanism to approximate the optimal Bayesian revenue using {\em finite} non-adaptive queries to such distributions,
it is intuitive that
some kind of {\em small-tail assumption} for the distributions
is needed.
Indeed, given any mechanism with query complexity $C$, there always exists a distribution that
has a sufficiently small probability mass around a sufficiently large value, such that the mechanism cannot find it using
$C$ queries.
If this probability mass is where all the revenue comes (e.g., all the remaining probability mass is around value 0),
then the mechanism cannot be a good approximation to $OPT$.
Following the literature  \cite{roughgarden2016ironing, devanur2016sample},
the small-tail assumptions are such that
the expected revenue generated
from the ``tail'' of the distributions
is negligible compared to the optimal revenue; see Section \ref{sec:small-tail}. Distributions with bounded supports 
automatically satisfy these assumptions, so are regular distributions in single-item auctions.

Even with small-tail assumptions,
it is hard to generate good revenue from unbounded distributions with finite {\em value} queries.
Instead, we show it is sufficient to use only
{\em quantile} queries.
As before, the seller uses our quantile-query algorithm $\cA_{Q}$ (defined in Section \ref{subsec:Aq})
to learn a distribution $\cD'$ that approximates $\cD$, and then reduces to simple mechanisms under $\cD'$.
However, even for single-item auctions, it is not so simple to show why the combination of these two parts work.
Indeed, under value queries it is easy to ``under-price'' the item so that the probability of sale is the same
as in the optimal mechanism for $\cD$.
Under quantile queries, under-pricing may lose a large amount of revenue because, for given quantiles,
there is no guarantee on where the corresponding values are.
Instead, the main idea in using quantile queries is to {\em ``over-price''} the item.
This is risky in many auction design scenarios, because it may significantly reduce the probability of sale, and thus lose a lot of revenue.
We prove a key technical lemma in Lemma \ref{lem:single:discretized} for single-item auctions, 
where we show that by discretizing the quantile space properly,
we can over-price the item while almost preserving the probability of sale as in the optimal mechanism under $\cD$.
In Lemma \ref{lem:additive:bvcg}, we prove another technical lemma showing that proper over-pricing can also be done in additive auctions.

Note that we can get the median of a distribution simply by querying the quantile $1/2$.
Then, for single-parameter auctions with regular distributions, using the parametric mechanism in \cite{azar2013optimal} we get the same revenue as theirs.
However, our query mechanisms can handle multi-parameter auctions and irregular distributions.

\subsection{Small-Tail Assumptions}
\label{sec:small-tail}

A Bayesian auction instance $\cI$ satisfies the {\em Small-Tail Assumption 1} if
there exists a function%
\footnote{If computation complexity is a concern, then one can further
require that the function is efficiently computable.} $h: (0,1) \to (0,1)$ such that,
for any constant $\delta_1 \in (0,1)$ and any BIC mechanism $\cM$,
letting $\epsilon_1 = h(\delta_1)$, we have
\begin{equation}\label{eq:tail:1}
\mathop\bE\limits_{v \sim \cD} {\bf I}_{\exists i,j, q_{ij}(v_{ij}) \leq \epsilon_1}
Rev(\cM(v;\cI))
\leq \delta_1 OPT(\cI).
\end{equation}
Here
$q_{ij}(v_{ij})$ is the quantile of $v_{ij}$ under distribution $\cD_{ij}$,
$Rev(\cM(v;\cI))$ is the revenue of $\cM$ under the Bayesian instance $\cI$ when the true valuation profile is $v$, and ${\bf I}$ is the indicator function.
For discrete distributions, Equation~\ref{eq:tail:1}
is imposed on the $\epsilon_1$ probability mass over the highest values.

Equation \ref{eq:tail:1}
immediately implies the following weaker {\em Small-Tail Assumption 2}:
there exists a
function $h: (0,1) \to (0,1)$ such that,
for any constant $\delta_1 \in (0,1)$ ,
letting $\epsilon_1 = h(\delta_1)$, we have
\begin{equation}\label{eq:tail:2}
\mathop\bE\limits_{v \sim \cD} {\bf I}_{\exists i,j, q_{ij}(v_{ij}) \leq \epsilon_1}
Rev_{OPT}(v;\cI) \leq \delta_1 OPT(\cI).
\end{equation}
Here $Rev_{OPT}(v; \cI)$ is the revenue generated by the
optimal BIC mechanism for $\cI$
when the true valuation profile is $v$.
Similar assumptions are widely adopted in sampling mechanisms
to deal with irregular distributions with unbounded supports.

\subsection{The Quantile-Query Algorithm}
\label{subsec:Aq}

We define our quantile-query algorithm $\cA_{Q}$ in Algorithm \ref{QueryProcedure:quantile}.
As before, $D\in \Delta(\bR)$ is the distribution to be queried.
The algorithm takes two parameters, the tail length~$\epsilon_1$ and the precision factor $\delta$,
makes $O(\log_{1+\delta}\frac{1}{\epsilon_1})$ quantile queries to the oracle, and
then returns a discrete distribution $D'$.

\begin{algorithm}[htbp]
  \caption{\hspace{-3pt} The Quantile-Query Algorithm $\cA_{Q}$}
 \label{QueryProcedure:quantile}
   \begin{algorithmic}[1]
\REQUIRE  the tail length $\epsilon_1$ and the precision factor $\delta$.

\STATE Let $k = \lceil \log_{1+\delta}\frac{1}{\epsilon_1}\rceil$ and define the {\em quantile vector} as
$q=(q_{0},q_{1},\dots, q_{k-1}, q_{k})
= (1,\epsilon_1(1+\delta)^{k-1}, \dots, \epsilon_1(1+\delta), \epsilon_1)$.\\

\STATE Query the oracle for $D$ with $q$, and receive a non-decreasing value vector
$(v_l)_{l\in \{0, \dots, k\}}$.\\

\STATE Construct a distribution $D'$ as follows:
$D'(v_{l}) = q_l - q_{l+1}$ for each $l\in \{0, \dots, k\}$,
where $q_{k+1}\triangleq 0$. \\

\ENSURE Distribution $D'$.
 \end{algorithmic}
\end{algorithm}

\subsection{Single-Item Auctions}
\label{sec:single}

Mechanism \ref{alg:single:rep} defines our {\em efficient quantile Myerson} mechanism $\cM_{EQM}$.

\begin{algorithm}[htbp]
\floatname{algorithm}{Mechanism}
  \caption{\hspace{-3pt} Efficient Quantile Myerson Mechanism $\cM_{EQM}$}
 \label{alg:single:rep}
\begin{algorithmic}[1]
\STATE Given $\epsilon > 0$, run algorithm $\cA_{Q}$ with $\delta = \frac{\epsilon}{3}$ and $\epsilon_{1} = h(\frac{2\epsilon}{3(1+\epsilon)})$ (i.e., $\delta_1 = \frac{2\epsilon}{3(1+\epsilon)}$ for Small Tail Assumption 2),
for each player $i$'s distribution~$\cD_{i}$. Denote by $\cD'_i$ the returned distribution.
Let $\cD'=\times_{i\in N}\cD'_{i}$.

\STATE Run Myerson's mechanism $\cM_{MRS}$ with $\cD'$ and the players' reported values, $b=(b_{i})_{i\in N}$,
to get allocation $x = (x_{i})_{i\in N}$ and price profile $p = (p_i)_{i\in N}$ as the outcome.
\end{algorithmic}
\end{algorithm}

\begin{theorem}
\label{thm:single}
$\forall \epsilon > 0$, any single-item instance $\cI = (N,M,\cD)$
satisfying Small-Tail Assumption~2,
 $\cM_{EQM}$ is DSIC, has query complexity
$O(- n\log_{1+\frac{\epsilon}{3}} {h(\frac{2\epsilon}{3(1+\epsilon)})})$, and
$Rev(\cM_{EQM}(\cI)) \geq \frac{OPT(\cI)}{1+\epsilon}$.
\end{theorem}

Before proving Theorem \ref{thm:single},
we first claim the following key lemma via an imaginary Bayesian mechanism that ``over-prices''.
Recall $\cI' = (N, M, \cD')$ is the instance under $\cD'$.

\begin{lemma}\label{lem:single:discretized}
$Rev(\cM_{MRS}(\cI')) \geq \frac{1}{1+\epsilon}OPT(\cI).$
\end{lemma}
\begin{proof}
For each player $i$, denote the support of $\cD'_i$ by $V'_i = (v'_{i;l})_{l\in \{0, \dots, k\}}$.
We first define a way to couple the values $v'_i\sim \cD'_i$ with the values $v_i\sim \cD_i$.

\vspace{-10pt}
\subparagraph{The randomized round-down scheme.}
For any value $v_i\geq v'_{i;0}$,
let $v_i^{-}$ be $v_i$ {\em rounded down} to the support of $\cD'_i$, such that $v_i^-$ is distributed according to $\cD'_i$ whenever
$v_i$ is distributed according to $\cD_i$.
Recall that
under value queries, $v_i^-$ is simply the largest value in $V'_i$ that is less than or equal to $v_i$,
no matter whether $\cD_i$ is continuous or discrete.
Under quantile queries, when $\cD_i$ is continuous, the same deterministic round-down scheme still works.
However, the situation is more subtle when $\cD_i$ is discrete,
 and we
need a randomized round-down scheme to ensure the relationship between $v_i$ and $v_i^-$.
More precisely, by the definition of quantile queries,
$V'_i$ is a subset of $\cD_i$'s support.
If $v_i$ is not in $V'_i$, then it is still deterministically rounded down as before.
If $v_i$ is in~$V'_i$, say $v_i = v'_{i;l}$, then by the definition of quantile queries and the construction of $\cD'_i$, we have
$\Pr_{x\sim \cD_i}[x\geq v_i] \geq q_l = \Pr_{x\sim \cD'_i}[x\geq v_i]$.
In this case, $v_i$ is rounded down to $v'_{i; l-1}$ (i.e., $v_i^- = v'_{i; l-1}$)
with probability
$$\frac{\Pr_{x\sim \cD_i}[x\geq v_i] - \Pr_{x\sim \cD'_i}[x\geq v_i]}{\cD_i(v_i)},$$
and to $v'_{i; l}$ (i.e.,  $v_i^- = v'_{i; l}$) with probability
$$1-\frac{\Pr_{x\sim \cD_i}[x\geq v_i] - \Pr_{x\sim \cD'_i}[x\geq v_i]}{\cD_i(v_i)}.$$
Following this scheme, it is not hard to verify that
$\Pr_{v_i\sim \cD_i}[v_i^- \geq v'_{i;l}] = q_l$ for any $l\in \{0, \dots, k\}$, thus $v_i^-$ is distributed according to $\cD'_i$, as desired.

No matter what $v_i^-$ is, let $v_i^+$ be
the smallest value in $V'_{i}$ that is strictly larger than $v_i^-$
(if no such value exists, then $v_i^{+}=+\infty$).
That is,
$v_i^+\geq v_i$ and $v_i^+$ is $v_i$ ``rounded up'', which was not needed under value queries and is new for quantile queries.

\vspace{-10pt}
\subparagraph{The randomized resampling scheme.}
For any value $v'_i\sim \cD'_i$, let $v_i$ be {\em resampled} from $\cD_i$
conditional on ``$v_i$ rounded down to~$v'_i$'',
so that $v_i$ is distributed according to $\cD_i$ whenever $v'_i$ is distributed
according to $\cD'_i$.
Again, under value queries, the resampling
is simply conditional on $v_i \in [v'_{i; l}, v'_{i;l+1})$ when $v'_i = v'_{i; l}$,
no matter whether $\cD_i$ is continuous or discrete.
Under quantile queries, this resampling scheme still works when $\cD_i$ is continuous.
When $\cD_i$ is discrete,
we need to ``undo'' the randomized round-down scheme defined above.
More precisely, letting $v'_i = v'_{i;l}$, $v_i$ is set to be $v'_{i;l+1}$
with probability
$$p_1 = \frac{\Pr_{x\sim \cD_i}[x \geq v'_{i;l+1}] - q_{l+1}}{\cD'_i(v'_{i;l})};$$
is resampled from $\cD_i$ conditional on $v_i \in (v'_{i; l}, v'_{i;l+1})$
with probability
$$p_2 = \frac{\Pr_{x\sim \cD_i}[v'_{i;l} < x < v'_{i;l+1}]}{\cD'_i(v'_{i;l})};$$
and is set to be $v'_{i;l}$
with probability
$$p_3 = \frac{\Pr_{x\sim \cD_{i}}[x\leq v'_{i;l}] - \Pr_{x\sim \cD'_{i}}[x < v'_{i;l}]}{\cD'_i(v'_{i;l})}
= \frac{\cD_i(v'_{i;l}) - \Pr_{x\sim \cD_i}[x \geq v'_{i;l}] + q_l}{\cD'_i(v'_{i;l})}.$$
Following this resampling scheme, it is not hard to verify that $v_i$ is distributed according to $\cD_i$
whenever $v'_i$ is distributed according to $\cD'_i$.

Given the round-down and the resampling schemes above,
we consider the Bayesian mechanism~$\cM^{*}$ defined in Mechanism \ref{alg:imaginary} for~$\cI'$, and compare its revenue with that of $\cM_{MRS}$.
We first claim that $\cM^*$ is a DSIC mechanism.
Because $\cM_{MRS}$ is DSIC,
each $x_i$ is monotone in $v_i$.
Although $v_i$ is a random variable given~$v'_i$,
it is easy to see that for any two different values $v'_i\in V'_i$ and $\hat{v}'_i\in V'_i$,
the corresponding resampled values $v_i$ and $\hat{v}_i$
are such that $v'_i< \hat{v}'_i$ implies
$v_i\leq  \hat{v}_i$.
Thus $x_i$ is monotone in $v'_i$ as well.
Moreover, let $\theta_i$ be player $i$'s threshold payment
in $\cM_{MRS}$ given $v_{-i}$ and $\cD$.
If $v'_i> \theta_i$ then $v_i> \theta_i$,
thus player $i$ gets the item at price $p_i = \theta_i$.
If $v'_i< \theta_i$ then player $i$ does not get the item and $p_i = 0$, no matter whether $v_i<\theta_i$ or not.
Accordingly, $\theta_i$ is also player $i$'s threshold payment in $\cM^*$ under $v_{-i}$ and $\cD'$.
Since $v_{-i}$ does not depend on $v'_i$,
$\cM^*$ is DSIC as desired.

\begin{algorithm}
\floatname{algorithm}{Mechanism}
  \caption{\hspace{-3pt} A Bayesian mechanism ${\cM}^{*}$ for instance $\cI'$}
  \label{alg:imaginary}
 \begin{algorithmic}[1]
\STATE Each player $i$ reports his value $v'_{i}$, and the mechanism discards the report that is not in $V'_i$.

   \STATE\label{step:unbounded:single:1}
   For each player $i$,
   generate value $v_{i}$ according to $v'_i$ using our resampling scheme.

 \STATE Run $\cM_{MRS}$ with the value profile $v$ and the prior distribution $\cD$, to get the price $p_i$ and the allocation $x_i\in \{0, 1\}$ for each player $i$.

  \STATE If $x_i = 1$ and $p_{i}\leq v'_{i}$, sell the item to $i$ and charge him $p_{i}$; otherwise, set $x_i = 0$ and $p_i=0$.
 \end{algorithmic}
\end{algorithm}

To analyze the revenue of $\cM^*$,
note that by construction,
when $v'_{i}$ is distributed according to~$\cD'_{i}$,
the resampled $v_{i}$ in $\cM^*$ is distributed according to $\cD_{i}$.
Moreover, each $v'_i$ is distributed as if we first sample $v_i$ from $\cD_i$ and then setting $v'_i = v_i^-$.

Thus, mechanism $\cM^{*}$ on instance $\cI'$ essentially generates the same expected  revenue as $\cM_{MRS}$ on instance $\cI$,
except for the case when $v'_{i} < p_{i} \leq v_{i}$ for the winner $i$.
Fortunately, we are able to upper-bound the probability of this event and thus upper-bound the expected revenue loss.
More precisely, for each player $i$, we write $p_i$ as $p_i(v_{-i}; \cD)$ to emphasize that it is the threshold
payment for~$i$ given $v_{-i}$ and $\cD$, and does not depend on $v_i$ or $v'_i$.
We have
\begin{eqnarray}
&& Rev(\cM^{*}(\cI'))
=\sum\limits_{i}
\mathop\mathbb{E}\limits_{v_{-i}\sim \cD_{-i}}
\mathop\mathbb{E}\limits_{v_{i}\sim \cD_{i}}
p_{i}(v_{-i}; \cD) {\bf I}_{v_{i}^-\geq p_{i}(v_{-i}; \cD)} \nonumber \\
&=& \sum\limits_{i}
\mathop\mathbb{E}\limits_{v_{-i}\sim \cD_{-i}}
p_{i}(v_{-i}; \cD) \cdot \Pr_{v_i\sim \cD_i}[v^-_{i}\geq p_{i}(v_{-i}; \cD)]. \hspace{15pt} \label{eqn:3}
\end{eqnarray}
Here the first equality holds because of the relationship between $\cD'$ and $\cD$ as established by our rounding and resampling schemes,
and because each player $i$ in $\cM^*$ pays the same threshold price as in
mechanism $\cM_{MRS}$
whenever $v'_{i}$ is at least the threshold, and pays 0 otherwise.
By the construction of the distribution~$\cD'$, we have the following claim, 
which is proved in Appendix \ref{app:unbounded}. 
\begin{claim}\label{clm:rev:gap}
$\Pr_{v_i\sim \cD_i}[v_{i}\geq p_{i}(v_{-i}; \cD) | {q_{i}(v_{i}) > \epsilon_1}]
\leq (1+\delta) \Pr_{v_i\sim \cD_i}[v^-_{i}\geq p_{i}(v_{-i}; \cD)].$
\end{claim}
Combining Equation \ref{eqn:3}, Claim \ref{clm:rev:gap} and Small Tail Assumption 2, we are able to lower-bound  the revenue of $\cM^{*}$ as follows,
which is also proved in Appendix \ref{app:unbounded}. 

\begin{claim}\label{clm:rev:m*}
$Rev(\cM^*(\cI')) \geq \frac{1}{1+\epsilon}OPT(\cI)$.
\end{claim}

By the optimality of $\cM_{MRS}$, $Rev(\cM_{MRS}(\cI')) \geq Rev(\cM^{*}(\cI'))$, and
Lemma \ref{lem:single:discretized} holds.
\end{proof}
\begin{proof} [Proof of Theorem \ref{thm:single}]
First, mechanism $\cM_{EQM}$ is DSIC because  $\cM_{MRS}$ is DSIC.
Second, the query complexity of $\cM_{EQM}$ is
$O(- n\log_{1+\frac{\epsilon}{3}} {h(\frac{2\epsilon}{3(1+\epsilon)})})$, because
there are $k+1 = \lceil\log_{1+\frac{\epsilon}{3}} \frac{1}{h(\frac{2\epsilon}{3(1+\epsilon)})}\rceil+1$
quantile queries for each player and there are $n$ players in total.
By definition,

$Rev(\cM_{EQM}(\cI))=Rev(\cM_{MRS}(\cI;\cD'))$.
By construction,  $\cD'$ is stochastically dominated by $\cD$.
Thus by revenue monotonicity
$Rev(\cM_{MRS}(\cI;\cD')) \geq Rev(\cM_{MRS}(\cI'))$.
Combining these two equations with
Lemma \ref{lem:single:discretized},
Theorem \ref{thm:single} holds.
\end{proof}

Mechanism $\cM_{EQM}$ and Theorem \ref{thm:single}
immediately extend to single-parameter downward-closed settings.
Finally, when the distributions are regular, we are able to prove an even better query complexity and a matching lower-bound; see Section~\ref{sec:single:regular}.

\subsection{Unit-Demand Auctions}\label{subsec:ua}

The unit-demand mechanism $\cM_{EQUD}$ is similar (see Mechanism~\ref{alg:unit}),
and we have the following.

\vspace{-5pt}
\begin{theorem}
\label{thm:unit-demand}
$\forall \epsilon > 0$, any unit-demand instance $\cI = (N,M,\cD)$
satisfying Small-Tail Assumption 2,
 $\cM_{EQUD}$ is DSIC, has
query complexity $O(- mn\log_{1+\frac{\epsilon}{3}} h(\frac{2\epsilon}{3(1+\epsilon)}))$, and
$Rev(\cM_{EQUD}(\cI)) \geq \frac{OPT(\cI)}{24(1+\epsilon)}$.
\end{theorem}

\begin{algorithm}[htbp]
\floatname{algorithm}{Mechanism}
  \caption{\hspace{-3pt}  Mechanism ${\cM}_{EQUD}$ for Unit-Demand Auctions}
  \label{alg:unit}
  \begin{algorithmic}[1]
\STATE Given $\epsilon > 0$, run algorithm $\cA_{Q}$ with $\delta = \frac{\epsilon}{3}$ and $\epsilon_{1} = h(\frac{2\epsilon}{3(1+\epsilon)})$ (i.e., $\delta_1 = \frac{2\epsilon}{3(1+\epsilon)}$ for Small Tail Assumption 2),
for each player~$i$'s distribution $\cD_{ij}$ on each item $j$.
Denote by $\cD'_{ij}$ the returned distribution. Let $\cD'_{i}=\times_{j\in M}\cD'_{ij}$ and $\cD'=\times_{i\in N}\cD'_{i}$.

\STATE Run mechanism $\cM_{UD}$ with $\cD'$ and the players' reported values, $b=(b_{ij})_{i\in N,j\in M}$,
to get allocation $x = (x_{ij})_{i\in N, j\in M}$ and price profile $p = (p_i)_{i\in N}$ as the outcome.
  \end{algorithmic}
\end{algorithm}

The proof of Theorem \ref{thm:unit-demand} is similar to that of Theorem \ref{thm:unit-demand:value_query},
but Lemma \ref{lem:single:discretized} above is used instead of Lemma~5 of \cite{devanur2016sample},
and the round-down scheme is replaced by the randomized round-down scheme designed in the proof of Lemma \ref{lem:single:discretized}.
The details have been omitted.

\subsection{Additive Auctions}\label{subsec:add}

For additive auctions, we cannot use Small-Tail Assumption 2, because it does not imply that
the revenue loss on the tail by running $\cM_{BVCG}$ is much less than the revenue of the optimal mechanism.
To approximate $\cM_{BVCG}$, not only we need Small-Tail Assumption 1, but
we also approximate $\cD$ by running the quantile-query algorithm $\cA_Q$ with different parameters.
The resulting mechanism $\cM_{EQBVCG}$ is defined in Mechanism~\ref{mechanism:eqbvcg},
and the mechanism $\cM_{EQIM}$ simply replaces $\cM_{BVCG}$ with $\cM_{IM}$.
Again, 
in the final mechanism $\cM_{EQA}$ the seller randomly chooses between the two query mechanisms,
according to probabilities
defined in the analysis.
We have the following theorem.

\begin{algorithm}
\floatname{algorithm}{Mechanism}
  \caption{\hspace{-3pt} Mechanism ${\cM}_{EQBVCG}$ for Additive Auctions}
  \label{mechanism:eqbvcg}
  \begin{algorithmic}[1]
\STATE Given $\epsilon > 0$, run algorithm $\cA_{Q}$ with
$\delta = (1+\frac{\epsilon}{5})^{1/m}-1$
and $\epsilon_1 = h(\frac{\epsilon}{10(1+\epsilon)})$
(i.e., $\delta_1 = \frac{\epsilon}{10(1+\epsilon)}$ for Small Tail Assumption 1),
for each player $i$'s distribution $\cD_{ij}$ on each item $j$.
Denote by $\cD'_{ij}$ the returned distribution. Let $\cD'_{i}=\times_{j\in M}\cD'_{ij}$ and $\cD'=\times_{i\in N}\cD'_{i}$.

\STATE Run $\cM_{BVCG}$ with $\cD'$ and the players' reported values, $b=(b_{ij})_{i\in N, j\in M}$,
to get allocation $x = (x_{ij})_{i\in N, j\in M}$ and price profile $p = (p_i)_{i\in N}$ as the outcome.
\end{algorithmic}
\end{algorithm}

\vspace{-5pt}

\begin{theorem} \label{thm:eca}
$\forall \epsilon > 0$, any additive instance
$\cI = (N, M, \cD)$ satisfying Small-Tail Assumption~1,
$\cM_{EQA}$ is DSIC, has query complexity
$O(-m^2n\log_{1+\frac{\epsilon}{5}} h(\frac{\epsilon}{10(1+\epsilon)})$, and
$Rev({\cM}_{EQA}(\cI))
\geq \frac{OPT(\cI)}{8(1+\epsilon)}$.
\end{theorem}

Before proving Theorem \ref{thm:eca}, we first analyze mechanism $\cM_{EQBVCG}$, and we have the following.

\begin{lemma}
\label{lem:additive:bvcg}
$\forall \epsilon > 0$, for any additive Bayesian instance $\cI = (N,M,\cD)$
satisfying Small-Tail Assumption 1,
$\cM_{EQBVCG}$ is DSIC, has query complexity
$O(- m^2n\log_{1+\frac{\epsilon}{5}} h(\frac{\epsilon}{10(1+\epsilon)}))$, and
$$Rev(\cM_{EQBVCG}(\cI))
\geq  \frac{1}{1+\frac{\epsilon}{5}}
\left (Rev(\cM_{BVCG}({\cI})) - \frac{\epsilon}{10(1+\epsilon)} OPT(\cI) \right ).
$$
\end{lemma}

\begin{proof}
First, mechanism $\cM_{EQBVCG}$ is DSIC
because $\cM_{BVCG}$ is DSIC. The query complexity is also immediate.

We now focus on the revenue of this mechanism.
We explicitly write $\cM_{BVCG}({\cI}; \cD')$
to emphasize the fact that
the seller runs mechanism $\cM_{BVCG}$ on the true valuation profile $v\sim \cD$,
but uses $\cD'$ to compute the entry fees $e_i$.
Given a player $i$ and a valuation profile $v$,
$p_i(v_i, \cD_i,v_{-i})$ is the price for~$i$ under $\cD_i$:
that is,
$$p_i(v_i, \cD_i,v_{-i})
= {\bf I}_{\sum_{j:v_{ij}\geq \beta_{ij}} ({v}_{ij}-\beta_{ij}) \geq e(\cD_i,v_{-i})}
(e(\cD_i,v_{-i}) + \sum_j \beta_{ij} {\bf I}_{v_{ij}\geq \beta_{ij}}),$$
where we omit $v_{-i}$ in $\beta_{ij}(v_{-i})$
when $v_{-i}$ is clear from the context.
\footnote{If there are ties in the players' values, then we distinguish between $\beta^+_{ij}$ and $\beta^-_{ij}$,
depending on the identity of the player with the highest bid for $j$ in $N\setminus\{i\}$.}
The price $p_i(v_i, \cD'_i, v_{-i})$ is similarly defined.
By the definition of the mechanism, we have

\begin{equation}
\label{eq:additive:ecbvcg:def}
Rev(\cM_{EQBVCG}(\cI)) = Rev({\cM}_{BVCG}({\cI}; \cD'))
= \sum_{i} \mathop\mathbb{E}\limits_{v_{-i}\sim \cD_{-i}} \mathop\mathbb{E}\limits_{v_i\sim \cD_i} p_i(v_i,\cD'_i,v_{-i}).
\end{equation}

Next, let $V'_{ij}$ be the support of $\cD'_{ij}$,
$V'_i = \times_{j\in M} V'_{ij}$,
round $v_{i}$ down to
the closest valuation $v'_{i}$ in $V'_{i}$ and compare the two valuation profiles
$(v'_i, v_{-i})$ and $(v_i, v_{-i})$.
By definition, $v'_{ij} \geq \beta_{ij}$ implies $v_{ij} \geq \beta_{ij}$.
Moreover, the entry fee of $i$ is the same under both valuation profiles,
as it only depends on
$\cD'_i$ and $v_{-i}$.
Similarly, the reserve price $\beta_{ij}$ is the same for any item $j$.
Thus we have
$e(\cD'_i,v_{-i}) + \sum_j \beta_{ij} {\bf I}_{v_{ij}\geq \beta_{ij}}
\geq e(\cD'_i,v_{-i}) + \sum_j \beta_{ij} {\bf I}_{v'_{ij}\geq \beta_{ij}}$
and ${\bf I}_{\sum_{j:v_{ij}\geq \beta_{ij}} (v_{ij}-\beta_{ij}) \geq e(\cD_i,v_{-i})}
\geq {\bf I}_{\sum_{j:v_{ij}\geq \beta_{ij}} (v'_{ij}-\beta_{ij}) \geq e(\cD'_i,v_{-i})}$.
Therefore

\begin{equation}\label{eq:add:1}
\mathop\mathbb{E}\limits_{v_{-i}\sim \cD_{-i}}   \mathop\mathbb{E}\limits_{v_i\sim \cD_i} p_i(v_i,\cD'_i,v_{-i})
\geq \mathop\mathbb{E}\limits_{v_{-i}\sim \cD_{-i}} \mathop\mathbb{E}\limits_{v_i\sim \cD_i} p_i(v'_i,\cD'_i,v_{-i})
=  \mathop\mathbb{E}\limits_{v_{-i}\sim \cD_{-i}} \mathop\mathbb{E}\limits_{v'_i\sim \cD'_i} p_i(v'_i,\cD'_i,v_{-i}),
\end{equation}
where the equality is again because drawing $v_{i}$ from $\cD_{i}$ and then rounding down to $v'_{i}$
is equivalent to drawing $v'_{i}$ from $\cD'_{i}$ directly.

In mechanism $\cM_{BVCG}$,
given $v_{-i}$ and $\cD'_{i}$,
$e(\cD'_{i},v_{-i})$ is the optimal entry fee
for maximizing the expected revenue generated from $i$,
where the expectation is taken over $\cD'_i$.
Accordingly,
\begin{equation}\label{eq:bvcg:opt}
\mathop\mathbb{E}\limits_{v'_i\sim \cD'_i} p_i(v'_i,\cD'_i,v_{-i})
\geq
\mathop\mathbb{E}\limits_{v'_i\sim \cD'_i} p_i(v'_i,\cD_i,v_{-i}).
\end{equation}
Combining Equations \ref{eq:additive:ecbvcg:def}, \ref{eq:add:1} and \ref{eq:bvcg:opt},
we have
\begin{equation}\label{equ:upperbvcg}
Rev(\cM_{EQBVCG}(\cI)) \geq
\sum\limits_{i} \mathop\mathbb{E}\limits_{v_{-i}\sim \cD_{-i}} \mathop\mathbb{E}\limits_{v'_i\sim \cD'_i} p_i(v'_i,\cD_i,v_{-i}).
\end{equation}
Thus  we  will use
$\sum\limits_{i} \mathop\mathbb{E}\limits_{v_{-i}\sim \cD_{-i}} \mathop\mathbb{E}\limits_{v'_i\sim \cD'_i} p_i(v'_i,\cD_i,v_{-i})$
to upper-bound $Rev(\cM_{BVCG}({\cI}))$.

To do so, first, for any player $i$, item $j$ and value $v_{ij}$,
if $v_{ij}< v'_{ij;k}$ where $v'_{ij;k}$ is the largest value in $V'_{ij}$,
then
denote by $\overline{v}_{ij}$ the smallest value in $V'_{ij}$ that is strictly larger than $v_{ij}$;
otherwise, let $\overline{v}_{ij} = v_{ij}$.
Moreover, denote by $\underline{v}_{ij}$ the largest value in $V'_{ij}$ that is weakly smaller than $v_{ij}$.
The valuation $\overline{v}_i$ and $\underline{v}_i$ are defined correspondingly given $v_i$.
Then We have
\begin{eqnarray}\label{equ:add:bvcg:2}
&& Rev(\cM_{BVCG}({\cI}))
= \sum_i \mathop\mathbb{E}\limits_{v_{-i}\sim \cD_{-i}}
\mathop\mathbb{E}\limits_{v_i\sim \cD_i}
p_i(v_i,\cD_i,v_{-i})\nonumber\\
&=& \sum_i \mathop\mathbb{E}\limits_{v_{-i}\sim \cD_{-i}}
\mathop\mathbb{E}\limits_{v_i\sim \cD_i}
{\bf I}_{\sum_{j:v_{ij}\geq \beta_{ij}} (v_{ij}-\beta_{ij}) \geq e(\cD_i,v_{-i})}
\left(e(\cD_i,v_{-i}) + \sum_j \beta_{ij} {\bf I}_{v_{ij}\geq \beta_{ij}}\right)
\nonumber\\
&=&
\sum_i \mathop\mathbb{E}\limits_{v_{-i}\sim \cD_{-i}}
\mathop\mathbb{E}\limits_{v_i\sim \cD_i}
{\bf I}_{\forall j, q_{ij}(v_{ij}) > \epsilon_1}
{\bf I}_{\sum_{j:v_{ij}\geq \beta_{ij}} (v_{ij}-\beta_{ij}) \geq e(\cD_i,v_{-i})}
\left(e(\cD_i,v_{-i}) + \sum_j \beta_{ij} {\bf I}_{v_{ij}\geq \beta_{ij}}\right) \nonumber\\
&+& \hspace{-5pt} \sum_i \mathop\mathbb{E}\limits_{v_{-i}\sim \cD_{-i}}
\mathop\mathbb{E}\limits_{v_i\sim \cD_i}
\hspace{-5pt} {\bf I}_{\exists j, q_{ij}(v_{ij}) \leq \epsilon_1}
{\bf I}_{\sum_{j:v_{ij}\geq \beta_{ij}} (v_{ij}-\beta_{ij}) \geq e(\cD_i,v_{-i})}
\left(e(\cD_i,v_{-i}) + \sum_j \beta_{ij} {\bf I}_{v_{ij}\geq \beta_{ij}}\right).\nonumber\\
\end{eqnarray}
Below we upper-bound the last two lines in Equation \ref{equ:add:bvcg:2} separately.
For the first part, we have
\begin{eqnarray}\label{equ:bvcg:3}
& & \sum_i \mathop\mathbb{E}\limits_{v_{-i}\sim \cD_{-i}}
\mathop\mathbb{E}\limits_{v_i\sim \cD_i}
{\bf I}_{\forall j, q_{ij}(v_{ij}) > \epsilon_1}
{\bf I}_{\sum_{j:v_{ij}\geq \beta_{ij}} (v_{ij}-\beta_{ij}) \geq e(\cD_i,v_{-i})}
\left(e(\cD_i,v_{-i}) + \sum_j \beta_{ij} {\bf I}_{v_{ij}\geq \beta_{ij}}\right) \nonumber\\
&\leq&
\sum_i \mathop\mathbb{E}\limits_{v_{-i}\sim \cD_{-i}}
\mathop\mathbb{E}\limits_{v_i\sim \cD_i}
{\bf I}_{\forall j, q_{ij}(v_{ij}) > \epsilon_1}
{\bf I}_{\sum_{j:\overline{v}_{ij}\geq \beta_{ij}} (\overline{v}_{ij}-\beta_{ij}) \geq e(\cD_i,v_{-i})}
\left(e(\cD_i,v_{-i}) + \sum_j \beta_{ij} {\bf I}_{\overline{v}_{ij}\geq \beta_{ij}}\right) \nonumber\\
&=& \sum_{i} \mathop\mathbb{E}\limits_{v_{-i}\sim \cD_{-i}}
\sum_{u_i\in V'_{i} : \sum_{j:u_{ij}\geq \beta_{ij}}(u_{ij}-\beta_{ij}) \geq e(\cD_i,v_{-i})}
\hspace{-20pt} \Pr_{v_i\sim\cD_i}[\overline{v}_i = u_i]
\left(e(\cD_i,v_{-i}) + \sum_j \beta_{ij} {\bf I}_{u_{ij}\geq \beta_{ij}}\right).\nonumber\\
\end{eqnarray}
The inequality above is because $v_{ij}\leq \overline{v}_{ij}$ for each player $i$ and item $j$,
which implies \\
${\bf I}_{\sum_{j:v_{ij}\geq \beta_{ij}} (v_{ij}-\beta_{ij}) \geq e(\cD_i,v_{-i})}
\leq
{\bf I}_{\sum_{j:\overline{v}_{ij}\geq \beta_{ij}} (\overline{v}_{ij}-\beta_{ij}) \geq e(\cD_i,v_{-i})}$
and
$\sum_j \beta_{ij} {\bf I}_{v_{ij}\geq \beta_{ij}} \leq \sum_j \beta_{ij} {\bf I}_{\overline{v}_{ij}\geq \beta_{ij}}$.
\\
Next, by the definition of the quantile vector $q$,
for any $u_{ij}\in V'_{ij}$
 we have
\begin{equation*}
\Pr_{v_{ij} \sim \cD_{ij}}[\overline{v}_{ij} = u_{ij}]
\leq (1+\delta) \Pr_{v_{ij} \sim \cD_{ij}}[\underline{v}_{ij} = u_{ij}].
\end{equation*}
Indeed, when $u_{ij} = v'_{ij;0}$, $\Pr[v_{ij}<u_{ij}]=0< (1+\delta)(1-\epsilon_1(1+\delta)^{k-1})
= (1+\delta)(q_0-q_1) = (1+\delta) \Pr[v_{ij}\in [v'_{ij;0}, v'_{ij;1})]$.
When $u_{ij} = v'_{ij;l}$ with $0<l<k$,
$\Pr(v_{ij}\in [v'_{ij;l-1}, v'_{ij;l})) = q_{l-1}-q_l\leq (1+\delta)q_l - q_l = \delta q_l = (1+\delta)\delta q_{l+1}
=(1+\delta)((1+\delta)q_{l+1}-q_{l+1}) = (1+\delta)(q_l-q_{l+1}) = (1+\delta)\Pr[v_{ij}\in [v'_{ij;l}, v'_{ij;l+1})]$.
And when $u_{ij} = v'_{ij;k}$,
$\Pr[v_{ij}\in [v'_{ij;k-1}, v'_{ij;k}]] = q_{k-1}-q_k = \delta\epsilon_1 <\epsilon_1 = \Pr[v_{ij}\geq v'_{ij;k}]$.
Since all distributions are independent, for any $u_i\in V'_{i}$ we have
\begin{equation}\label{eq:prob}
\Pr_{v_{i} \sim \cD_{i}}[\overline{v}_{i} = u_i]
\leq (1+\delta)^m \Pr_{v_{i} \sim \cD_{i}}[\underline{v}_{i} = u_i].
\end{equation}
Combining Equations \ref{equ:bvcg:3} and \ref{eq:prob}, we have
\begin{eqnarray}\label{equ:bvcg:4}
& & \sum_i \mathop\mathbb{E}\limits_{v_{-i}\sim \cD_{-i}}
\mathop\mathbb{E}\limits_{v_i\sim \cD_i}
{\bf I}_{\forall j, q_{ij}(v_{ij}) > \epsilon_1}
{\bf I}_{\sum_{j:v_{ij}\geq \beta_{ij}} (v_{ij}-\beta_{ij}) \geq e(\cD_i,v_{-i})}
\left(e(\cD_i,v_{-i}) + \sum_j \beta_{ij} {\bf I}_{v_{ij}\geq \beta_{ij}}\right) \nonumber\\
&\leq&
\sum_{i} \mathop\mathbb{E}\limits_{v_{-i}\sim \cD_{-i}}
\sum_{u_i\in V'_{i} : \sum\limits_{j:u_{ij}\geq \beta_{ij}} \hspace{-5pt} (u_{ij}-\beta_{ij}) \geq e(\cD_i,v_{-i})}
\hspace{-.6in} (1+\delta)^m \cdot \Pr_{v_i\sim\cD_i}[\underline{v}_i = u_i] \cdot
\left(e(\cD_i,v_{-i}) + \sum_j \beta_{ij} {\bf I}_{u_{ij}\geq \beta_{ij}}\right) \nonumber \\
&=& (1+\delta)^m \sum_{i} \mathop\mathbb{E}\limits_{v_{-i}\sim \cD_{-i}} \mathop\mathbb{E}\limits_{v'_{i}\sim \cD'_{i}}
{\bf I}_{\sum_{j:v'_{ij}\geq \beta_{ij}} (v'_{ij}-\beta_{ij}) \geq e(\cD_i,v_{-i})}
\left(e(\cD_i,v_{-i}) + \sum_j \beta_{ij} {\bf I}_{v'_{ij}\geq \beta_{ij}}\right)
\nonumber \\
&=& (1+\delta)^m \sum\limits_{i} \mathop\mathbb{E}\limits_{v_{-i}\sim \cD_{-i}} \mathop\mathbb{E}\limits_{v'_{i}\sim \cD'_{i}}
 p_i(v'_i,\cD_i,v_{-i})
\leq (1+\delta)^m Rev(\cM_{ECBVCG}(\hat{\cI})).
\end{eqnarray}
The first equality above holds because drawing $v_{i}$ from $\cD_{i}$ and rounding down to the support of $\cD'_{i}$
is equivalent to drawing $v'_{i}$ from $\cD'_{i}$.
The second equality is by the definition of $p_i(v'_i,\cD_i,v_{-i})$,
and the last inequality holds by Equation \ref{equ:upperbvcg}.

By Equations \ref{equ:add:bvcg:2} and \ref{equ:bvcg:4}, we have
\begin{eqnarray}\label{equ:bvcg:5}
&& Rev(\cM_{BVCG}({\cI})) \nonumber \\
&\leq& (1+\delta)^m Rev(\cM_{EQBVCG}({\cI})) \nonumber \\
&+&
\hspace{-5pt} \sum_i \mathop\mathbb{E}\limits_{v_{-i}\sim \cD_{-i}}
\mathop\mathbb{E}\limits_{v_i\sim \cD_i}
\hspace{-5pt} {\bf I}_{\exists j, q_{ij}(v_{ij}) \leq \epsilon_1}
{\bf I}_{\sum_{j:v_{ij}\geq \beta_{ij}} (v_{ij}-\beta_{ij}) \geq e(\cD_i,v_{-i})}
\left(e(\cD_i,v_{-i}) + \sum_j \beta_{ij} {\bf I}_{v_{ij}\geq \beta_{ij}}\right).\nonumber\\
\end{eqnarray}

\noindent
For the last line of Equation \ref{equ:bvcg:5},
we have
\begin{eqnarray}\label{eq:bvcg:tail}
& & \sum_i \mathop\mathbb{E}\limits_{v_{-i}\sim \cD_{-i}}
\mathop\mathbb{E}\limits_{v_i\sim \cD_i}
\hspace{-5pt} {\bf I}_{\exists j, q_{ij}(v_{ij}) \leq \epsilon_1}
{\bf I}_{\sum_{j:v_{ij}\geq \beta_{ij}} (v_{ij}-\beta_{ij}) \geq e(\cD_i,v_{-i})}
\left(e(\cD_i,v_{-i}) + \sum_j \beta_{ij} {\bf I}_{v_{ij}\geq \beta_{ij}}\right) \nonumber \\
&=&
\mathop\mathbb{E}\limits_{v\sim \cD}
\sum_i {\bf I}_{\exists j, q_{ij}(v_{ij}) \leq \epsilon_1}
{\bf I}_{\sum_{j:v_{ij}\geq \beta_{ij}} (v_{ij}-\beta_{ij}) \geq e(\cD_i,v_{-i})}
\left(e(\cD_i,v_{-i}) + \sum_j \beta_{ij} {\bf I}_{v_{ij}\geq \beta_{ij}}\right) \nonumber \\
&\leq&
\mathop\mathbb{E}\limits_{v\sim \cD}
{\bf I}_{\exists i,j, q_{ij}(v_{ij}) \leq \epsilon_1}
\sum_i
{\bf I}_{\sum_{j:v_{ij}\geq \beta_{ij}} (v_{ij}-\beta_{ij}) \geq e(\cD_i,v_{-i})}
\left(e(\cD_i,v_{-i}) + \sum_j \beta_{ij} {\bf I}_{v_{ij}\geq \beta_{ij}}\right) \nonumber \\
&=& \mathop\mathbb{E}\limits_{v\sim \cD}
{\bf I}_{\exists i,j, q_{ij}(v_{ij}) \leq \epsilon_1} Rev(\cM_{BVCG}(v; \cI)) \leq  \frac{\epsilon}{10(1+\epsilon)} OPT(\cI).
\end{eqnarray}
The first inequality above is because, for each player $i$ and valuation profile $v$,
${\bf I}_{\exists j, q_{ij}(v_{ij}) \leq \epsilon_1} \leq {\bf I}_{\exists i,j, q_{ij}(v_{ij}) \leq \epsilon_1}$.
The second inequality is by
the Small-Tail Assumption 1.

Combining Equations \ref{equ:bvcg:5} and \ref{eq:bvcg:tail}, we have
$$Rev(\cM_{BVCG}({\cI}))
\leq (1+\delta)^m Rev(\cM_{EQBVCG}(\cI)) + \frac{\epsilon}{10(1+\epsilon)} OPT(\cI).$$
By the construction of Mechanism \ref{mechanism:eqbvcg},
$(1+\delta)^m = 1+\frac{\epsilon}{5}$.
Therefore  Lemma \ref{lem:additive:bvcg} holds.
\end{proof}

\begin{proof}
[Proof of Theorem \ref{thm:eca}]
First, as both $\cM_{EQBVCG}$ and $\cM_{EQIM}$ are DSIC, $\cM_{EQA}$ is DSIC.
Second, note that $\cM_{EQA}$ runs both mechanisms with
$\delta = (1+\frac{\epsilon}{5})^{1/m}-1$ and $\epsilon_1 = h(\frac{\epsilon}{10(1+\epsilon)})$.
To ease the analysis, when running mechanism $\cM_{EQIM}$, let
$\delta = \frac{\epsilon}{15}$ and $\epsilon_1 = h(\frac{2\epsilon}{3(5+\epsilon)})$:
 that is, set $\epsilon' = \frac{\epsilon}{5}$ and run mechanism $\cM_{EQM}$ with parameter $\epsilon'$ for
 each item.
By Theorem \ref{thm:single},
with $O(- mn\log_{1+\frac{\epsilon}{15}} h(\frac{2\epsilon}{3(5+\epsilon)}))$ queries,
$$Rev(\cM_{EQIM}(\cI)) \geq \frac{1}{1+\frac{\epsilon}{5}} Rev(\cM_{IM}(\cI)).$$
By Lemma \ref{lem:additive:bvcg},
with $O(- m^2n\log_{1+\frac{\epsilon}{5}} h(\frac{\epsilon}{10(1+\epsilon)}))$ queries,
$$Rev(\cM_{EQBVCG}(\cI))
\geq \frac{1}{1+\frac{\epsilon}{5}} \left(Rev(\cM_{BVCG}(\cI)) - \frac{\epsilon}{10(1+\epsilon)} OPT(\cI)\right).$$
Note that the total query complexity is still $O(- m^2n\log_{1+\frac{\epsilon}{5}} h(\frac{\epsilon}{10(1+\epsilon)}))$.

Let mechanism $\cM_{EQA}$ run
$\cM_{EQBVCG}$ with probability $\frac{1}{4}$ and
$\cM_{EQIM}$ with probability $\frac{3}{4}$.
We have
\begin{eqnarray*}
&& Rev({\cM}_{EQA}(\cI))
= \frac{1}{4} Rev({\cM}_{EQBVCG}(\cI)) + \frac{3}{4} Rev({\cM}_{EQIM}(\cI))\\
&\geq&\frac{1}{4(1+\frac{\epsilon}{5})} \left(Rev(\cM_{BVCG}(\cI)) - \frac{\epsilon}{10(1+\epsilon)} OPT(\cI)\right) +
\frac{3}{4(1+\frac{\epsilon}{5})} Rev(\cM_{IM}(\cI)) \\
&\geq& \frac{1}{1+\frac{\epsilon}{5}}
\left(  \frac{1}{4}Rev(\cM_{BVCG}(\cI)) +
 \frac{3}{4}Rev(\cM_{IM}(\cI))
 - \frac{\epsilon}{10(1+\epsilon)} OPT(\cI)\right) \\
&\geq& \frac{1}{1+\frac{\epsilon}{5}}
\left( \frac{1}{8}OPT(\cI)
 - \frac{\epsilon}{10(1+\epsilon)} OPT(\cI)\right)
= \dfrac{1}{8(1+\epsilon)} OPT(\cI).
\end{eqnarray*}

The last inequality above holds because $2\cM_{BVCG}(\cI) + 6\cM_{IM}(\cI) \geq OPT(\cI)$ \cite{cai2016duality}.
Thus Theorem~\ref{thm:eca} holds.
\end{proof}

The main advantage of using quantile queries is to handle unbounded distributions.
In addition, we can use the resulting query mechanisms to construct {\em sampling} mechanisms;
see Section~\ref{sec:app:sample}.
As shown in Theorem \ref{thm:eca},
the query complexity of mechanism $\cM_{EQA}$ has an extra factor of~$m$ compared with that of $\cM_{EVA}$ (and the lower bound).
It would be interesting to see whether our lower-bounds can be improved in this scenario.

\subsection{Using Quantile Queries for Bounded Distributions}
\label{sec:bounded:quantile}

As a corollary, Theorems \ref{thm:single}, \ref{thm:unit-demand} and \ref{thm:eca}
also provide another way to approximate the optimal BIC revenue using only quantile queries when the distributions are bounded.
More precisely, we have the following.
\begin{corollary}
\label{coro:quantile:bounded}
For any $\epsilon > 0$, $H>1$,
and prior distribution $\cD$
with each $\cD_{ij}$ bounded within $[1,H]$,
there exist DSIC mechanisms that use
$O(mn\log_{1+\epsilon} \frac{nmH(1+\epsilon)}{\epsilon})$ quantile queries for single-item auctions and unit-demand auctions,
and use $O(m^2n\log_{1+\epsilon} \frac{nmH(1+\epsilon)}{\epsilon})$ quantile queries for additive auctions,
whose approximation ratios to $OPT$ are respectively $1+\epsilon$, $24(1+\epsilon)$ and $8(1+\epsilon)$.
\end{corollary}
\begin{proof}
We only need to show that the Small-Tail Assumptions 1 and 2 are naturally satisfied when the distributions have bounded supports.
For example, consider additive auctions where all values are in $[1, H]$, as considered in \cite{huang2015making, cole2014sample}.
Then $mH$ and $1$ are straightforward upper- and lower-bounds for $OPT(\cI)$, respectively.
Moreover, by individual rationality,
$mH$ is an upper-bound for the revenue generated under any valuation profiles.
Given $\delta_1$, let $\epsilon_{1} = h(\delta_1) =\frac{\delta_1}{m^2nH}$ and denote by $E$ the event that there exist
at least one player~$i$ and one item $j$ with $q_{ij}(v_{ij}) \leq \epsilon_{1}$. By the union bound,
$\Pr[ E ]\leq mn \epsilon_1 = mn \cdot \frac{\delta_1}{m^{2}nH} = \frac{\delta_1}{mH}$.
Therefore
$$
\mathop\bE\limits_{v \sim \cD} {\bf I}_{\exists i,j, q_{ij}(v_{ij}) \leq \epsilon_1}
Rev(\cM(v;\cI))
\leq mH\cdot \Pr[ E ] \leq \delta_1 \leq \delta_1 OPT(\cI).
$$
Combining this observation with Theorems \ref{thm:single}, \ref{thm:unit-demand} and \ref{thm:eca},
we have Corollary \ref{coro:quantile:bounded} when the values are all bounded in $[1,H]$.
\end{proof}

\section{Single-Item Auctions with Regular Distributions}
\label{sec:single:regular}

In this section, we show that when we only consider regular distributions for single-item auctions,
the query complexity can be much lower.
In fact,
we no longer need the small-tail assumptions even when the supports are unbounded.
Here our lower- and upper-bounds are {\em tight} upto a logarithmic factor,
and require different techniques from previous sections.

For the lower-bound, recall that in Section \ref{sec:lower-bound} we allow the distributions to be irregular.
To construct the desired distributions, we can first find the
un-queried quantile interval and then move the probability mass from its end points to internal points.
Because the distributions can be irregular, we have complete control on where to put the probability mass.
However, if the distributions have to be regular then this cannot be done.
Instead, we start from two different single-peak revenue curves and construct regular distributions from them.
We still want to move probability mass from the end points of the un-queried quantile interval to internal points,
but such moves must be continuous in order to preserve regularity.

For the upper-bound, we show that regular distributions satisfy the
small-tail property with a properly defined tail function.
Thus our techniques for distributions with small-tails directly apply here.

\subsection{Lower Bound}

With regular distributions, by \cite{dhangwatnotai2015revenue} it is sufficient to use a single sample to achieve $2$-approximation
in revenue for {\em single-player} single-item auctions.
Because every distribution is a uniform distribution in the quantile
space,
a sample for such auctions
can be obtained by first choosing a quantile~$q$ uniformly at random from $[0,1]$ and then making a quantile query.
Thus, a single query is also sufficient for $2$-approximation in this case.
As such, unlike Theorem \ref{thm:lower:multi} where we have proved lower bounds for the query complexity for arbitrary constant approximations,
for regular distributions we consider lower bounds for $(1+\epsilon)$-approximations, where $\epsilon$ is sufficiently small.
More precisely, we have the following.

\begin{theorem}\label{thm:low:regular}
For any constant $\epsilon \in (0, \frac{1}{64})$, there exists a constant $C$ such that,
for any $n\geq 1$,
any DSIC Bayesian mechanism $\cM$
making less than $C n\epsilon^{-1}$
non-adaptive value and quantile queries to the oracle,
there exists a multi-player single-item Bayesian auction instance $\cI = (N, M, \cD)$
where $|N|=n$ and $\cD$ is regular, such that
$Rev(\cM(\cI))<
\frac{OPT(\cI)}{1+\epsilon}$.
\end{theorem}
We only prove Theorem \ref{thm:low:regular} for the single-player case, as in the following lemma.
The lower bound for general multi-player single-item auctions can be proved using the same technique as in Theorem~\ref{thm:lower:multi},
thus the full proof
has been omitted.

\begin{lemma}\label{lem:low:regular}
For any constant $\epsilon \in (0, \frac{1}{64})$, there exists a constant $C$ such that,
for any DSIC Bayesian mechanism $\cM$
making less than $C/\epsilon$
non-adaptive value and quantile queries to the oracle,
there exists a single-player single-item Bayesian auction instance $\cI = (N, M, \cD)$
where $\cD$ is regular, such that
$Rev(\cM(\cI))<
\frac{OPT(\cI)}{1+\epsilon}$.
\end{lemma}

\begin{proof}
Since the distributions are unbounded, we can always construct the distributions such that for any finite number of value queries,
the responses for the value queries have almost none contribution to the optimal revenue.
Thus we only need to focus on the lower bound for quantile queries.

Letting $k \triangleq \lceil \frac{1}{\des}\rceil$ and
$C \triangleq  \frac{1-2\des}{2\delta}$.
Here $\delta$ is a constant to be determined later and $\delta, \epsilon$ satisfies that $k\geq 2$.
In our construction, we divide the quantile interval $[0, 1]$
into $k+1$ sub-intervals each, with the right-end points defined as follows:
 from left to right,
$q_{0} = 0$, $q_{t+1} = q_{t}+\des$ for each $t\in \{0, \dots, k-1\}$.

Accordingly, for any Bayesian mechanism $\cM$ that makes less than
$\frac{C}{\epsilon}$ non-adaptive quantile queries, 
there exists a quantile interval $(q_t, q_{t+1})$ such that,
$q_{t+1} \leq 1-2\des$ and with probability at least $\frac{1}{2}$,
no quantile in $(q_{t}, q_{t+1})$ is queried. 
Indeed, if this is not the case,
then with probability at least $\frac{1}{2}$,
all the quantile intervals except $(1-2\des, 1-\des)$ and $(1-\des, 1)$ are queried.
Since there are at least $k-2$ quantile intervals,
the expected total number of queries made by $\cM$
is at least $\frac{k}{2}-1 \geq \frac{1-2\des}{2\des} = \frac{C}{\epsilon}$,
a contradiction.

We now construct two different
single-player single-item Bayesian instances
$$\{\cI_z = (N, M, \cD_z)\}_{z\in\{1,2\}},$$
where the distributions
outside the quantile range
$(q_t, q_{t+1})$ are all the same.
Thus with probability at least $\frac{1}{2}$,
mechanism $\cM$ cannot distinguish the $\cI_z$'s from each other.
We then show that
when this happens,
mechanism $\cM$ cannot be a $(1+3\epsilon)$-approximation for all instances $\cI_z$.

Let $R$ be a parameter that is large enough such that no value query will get any useful response.
Then the first distribution $\cD_1$ with value bounded within $[0,\frac{R}{q_{t}}]$ is defined as follows, where $F_{1}(\cdot)$ is the cumulative probability function of $\cD_1$.
\[F_{1}(v)=\left\{\begin{array}{ll}
1- \frac{R}{(1-q_{t+1})v+R}, & 0\leq v < \frac{R}{q_{t}},\\
1, & v=\frac{R}{q_{t}}.
\end{array}\right.\]
That is there is a probability mass $\frac{q_t}{1-\des}$ at value $\frac{R}{q_{t}}$
and within interval $[0, \frac{R}{q_{t}})$ it is a continuous distribution.
Then for any quantile in range $(0,\frac{q_t}{1-\des}]$,
the oracle will response $\frac{R}{q_t}$.
For quantile $q$ in range $(\frac{q_t}{1-\des}, 1]$,
the oracle will response $v(q) = \frac{R}{1-q_{t+1}}(\frac{1}{q}-1)$.
Therefore the revenue function with related to the quantile $q$ is
\[R_{1}(q)=\left\{\begin{array}{ll}
\frac{R}{1-q_{t+1}}(1-q), & \frac{q_t}{1-\des}< q \leq 1,\\
\frac{R}{1-\des}, & q=\frac{q_t}{1-\des}.
\end{array}\right.\]
The revenue curve $R_{1}(q)$ is illustrated figure \ref{fig:regular:1}.

\begin{figure}[htbp]
\vspace{10pt}
\begin{center}
\centering
\setlength{\unitlength}{1.3cm}
\thinlines
\begin{picture}(6,4.5)

\put(0,0){\vector(1,0){6}}
\put(0,0){\vector(0,1){4.2}}

\put(-0.4,2.5){$R$}
\put(-0.7,3.5){$\frac{R}{1-\des}$}
\put(0,4.4){$R_1(q)$}

\put(2.5,-0.3){$q_{t}$}
\put(3.3,-0.3){$\frac{q_t}{1-\des}$}
\put(4,-0.3){$q_{t+1}$}
\put(0,-0.3){$0$}
\put(5.4,-0.3){$1$}
\put(5.8,-0.3){$q$}

{\thicklines
\put(0,0){\line(1,1){3.6}}
\put(3.6,3.6){\line(1,-2){1.8}}
}

\multiput(3.6,0)(0,0.2){18}{\line(0,1){0.1}}

\multiput(0,3.6)(0.2,0){18}{\line(1,0){0.1}}
\multiput(0,2.6)(0.2,0){21}{\line(1,0){0.1}}

\multiput(2.6,0)(0,0.2){13}{\line(0,1){0.1}}
\multiput(4.1,0)(0,0.2){13}{\line(0,1){0.1}}

\end{picture}
\end{center}
\caption{The revenue curve of $\cD_1$.}
\label{fig:regular:1}
\end{figure}

The second distribution $\cD_2$ with value bounded within $[0,\frac{R}{q_{t}}]$ is defined as follows,
where $F_{2}(\cdot)$ is the cumulative probability function of $\cD_2$.
Let $v^{*}=\frac{R(2-\des)}{2(1-\des)-(2-\des)(1-q_{t+1})}$.
Since $q_{t+1} \leq 1-2\des$,
$v^{*}>0$ is well defined
and it is easy to check $v^{*}<\frac{R}{q_{t}}$.
\[F_{2}(v)=\left\{\begin{array}{ll}
1- \frac{R}{(1-q_{t+1})v+R}, & 0\leq v <v^{*} ,\\
1- \frac{R(1-\des)}{(1+q_{t}-\des)v-R}, & v^{*}\leq v < \frac{R}{q_{t}},\\
1, & v=\frac{R}{q_{t}}.
\end{array}\right.\]
That is, there is a probability mass $q_t$ at value $\frac{R}{q_{t}}$
and a two-step continuous distribution within $[q_{t},q^{*}]$ and $[q^{*},q_{t+1}]$.
Thus for any quantile in range $(0,q_t]$,
the oracle will response $\frac{R}{q_{t}}$.
It can be calculated that the quantile of value $v^{*}$ is $q^* = 1-\frac{2-\des}{2(1-\des)}\cdot (1-q_{t+1})$.
Then for quantile $q$ in range $(q_t, q^*]$,
the oracle will response $v(q) = \frac{R}{q}(1 - \frac{ q_t}{1+q_t-\des}) + \frac{R}{1+q_t-\des}$.
For quantile $q$ in range $(q^*, 1]$,
the oracle will response $v(q) = \frac{R}{1-q_{t+1}}(\frac{1}{q}-1)$.
Therefore the revenue function with related to the quantile $q$ is
\[R_{2}(q)=\left\{\begin{array}{ll}
\frac{R}{1-q_{t+1}}(1-q), & q^{*}< q \leq 1,\\
\frac{R}{1+q_{t}-\des}(1+q-\des), & q_{t}\leq q < q^{*},\\
R, & q = q_t.
\end{array}\right.\]
The revenue curve $R_{2}(q)$ is illustrated figure \ref{fig:regular:2}.

\begin{figure}[htbp]
\vspace{10pt}
\begin{center}
\centering
\setlength{\unitlength}{1.3cm}
\thinlines
\begin{picture}(6,4.1)

\put(0,0){\vector(1,0){6}}
\put(0,0){\vector(0,1){3.9}}

\put(-0.4,2.5){$R$}
\put(-1.6,3.1){$\frac{R}{2(1-\des)}+\frac{R}{2}$}
\put(0,4.1){$R_2(q)$}

\put(2.5,-0.3){$q_{t}$}
\put(3.6,-0.3){$q^*$}
\put(4,-0.3){$q_{t+1}$}
\put(-0.3,0){$0$}
\put(5.4,-0.3){$1$}
\put(5.8,-0.3){$q$}

{\thicklines
\put(0,0){\line(1,1){2.6}}
\put(2.6,2.6){\line(2,1){1.2}}
\put(3.8,3.2){\line(1,-2){1.6}}
}


\multiput(0,3.2)(0.2,0){19}{\line(1,0){0.1}}
\multiput(0,2.6)(0.2,0){21}{\line(1,0){0.1}}

\multiput(2.6,0)(0,0.2){13}{\line(0,1){0.1}}
\multiput(4.1,0)(0,0.2){13}{\line(0,1){0.1}}
\multiput(3.8,0)(0,0.2){16}{\line(0,1){0.1}}

\end{picture}
\end{center}
\caption{The revenue curve of $\cD_2$.}
\label{fig:regular:2}
\end{figure}

Indeed when the quantile query is from $[0, q_{t}] \cup [q_{t+1}, 1]$,
the oracle's answers for all distributions are the same.
Accordingly, with  probability at least $\frac{1}{2}$,
mechanism $\cM$
cannot distinguish $\cD_z$'s from each other,
which means it cannot distinguish
$\cI_z$'s from each other, as desired.

Since $\cM$ is truthful, the allocation rule for the player must be monotone
and he will pay the threshold payment set by $\cM$, denoted by $P$.
Let $P^{*}=\frac{(4-\des)R}{4(1-\des)-(4-\des)(1-q_{t+1})}$.
Here $P$ may be randomized.
Recall that $OPT(\cI_1) = \frac{R}{1-\des}$.
If with probability $\frac{1}{2}$ setting the price $P \leq P^{*}$, then for instance $\cI_1$, we have
\begin{eqnarray*}
&&Rev(\cM(\cI_1)) \leq \frac{1}{2} OPT(\cI_1) + \frac{1}{2}(\frac{3R}{4(1-\des)}+\frac{R}{4}) \\
&=& \frac{7R}{8(1-\des)} + \frac{R}{8}
=\frac{R}{1-\des}\left( 1-\frac{1}{8}\des\right)
< \frac{OPT(\cI_1)}{1+4\epsilon}
\end{eqnarray*}
when $\delta \geq 32$.
On the other hand, recall that $OPT(\cI_2) = \frac{R}{2(1-\des)}+\frac{R}{2} = \frac{(2-\des)R}{2(1-\des)}$.
If with probability $\frac{1}{2}$, the price $P > P^{*}$, for instance $\cI_2$,
we have
\begin{eqnarray*}
&&Rev(\cM(\cI_2)) < \frac{1}{2} OPT(\cI_2) + \frac{(4-\des)R}{2(4-2\des)}
=  \frac{(2-\des)R}{4(1-\des)} + \frac{(4-\des)R}{2(4-2\des)} \\
&=&\frac{(2-\des)R}{2(1-\des)} \left(\frac{1}{2}+\frac{(4-\des)(1-\des)}{2(2-\des)^{2}} \right)
=  OPT(\cI_2) \left(1 - \frac{\des}{2(2-\des)^{2}} \right)
< \frac{OPT(\cI_2)}{1+4\epsilon}
\end{eqnarray*}
when $\delta \geq 32$.
Thus for any mechanism $\cM$ with $O(\frac{1}{\epsilon})$ quantile queries,
there exists $z^{*}\in \{0, 1\}$ such that when $\epsilon < \frac{1}{64}$
and $\delta = 32$,
$$Rev(\cM(\cI_{z^{*}}))
\leq \frac{OPT(\cI_{z^*})}{2} + \frac{OPT(\cI_{z^*})}{2(1+4\epsilon)}
< \frac{OPT(\cI_{z^{*}})}{1+\epsilon}.$$
Therefore Lemma \ref{lem:low:regular} holds.
\end{proof}

\subsection{Upper Bound}

Our mechanism $\cM_{EMR}$ (i.e., ``Efficient quantile Myerson mechanism for Regular distributions'')
first constructs the distribution $\cD'$ that approximates $\cD$ using the quantile-query algorithm $\cA_Q$
with parameters
$\delta = \frac{\epsilon}{4}$ and $\epsilon_1 = \frac{\epsilon^2}{256n}$;
and then runs Myerson's mechanism $\cM_{MRS}$ on $\cD'$.
Formally, we have the following theorem.

\begin{theorem}\label{thm:single:regular}
$\forall \epsilon \in (0,1)$, and for any single-item instance $\cI = (N,M,\cD)$
where $\cD$ is regular, mechanism $\cM_{EMR}$ is DSIC, has query complexity $O(n\log_{1+\epsilon} \frac{n}{\epsilon})$,
and
$Rev(\cM_{EMR}(\cI)) \geq \frac{OPT(\cI)}{1+\epsilon}$.
\end{theorem}
\begin{algorithm}[htbp]
\floatname{algorithm}{Mechanism}
  \caption{\hspace{-3pt} Efficient quantile Myerson mechanism for regular distributions, $\cM_{EMR}$}
 \label{alg:single:regular}
\begin{algorithmic}[1]
\STATE Given $\epsilon > 0$, run algorithm $\cA_{Q}$ with $\delta = \frac{\epsilon}{4}$ and $\epsilon_{1} = \frac{\epsilon^2}{256n}$
for each player $i$'s distribution $\cD_{i}$, with the returned distribution denoted by $\cD'_{i}$.
Let $\cD'=\times_{i\in N}\cD'_{i}$.

\STATE Run $\cM_{MRS}$ with $\cD'$ and the players' reported values, $b=(b_{i})_{i\in N}$,
to get allocation $x = (x_{i})_{i\in N}$ and price profile $p = (p_i)_{i\in N}$ as the outcome.
\end{algorithmic}
\end{algorithm}

\begin{proof}
Consider the quantile value $q^* = \frac{\epsilon^2}{256n}$
and $v_i^* = F_i^{-1}(1-q^*)$.
Let $\hat{v}_i = \max\{v_i^*, \frac{16OPT(\cI)}{\epsilon}\}$,
and $\bar{\cD}_1, \dots, \bar{\cD}_n$ be imaginary distributions obtained by
truncating $\cD_1, \dots, \cD_n$ at $\hat{v}_i$
(i.e., a sample $\bar{v}_i$ from $\bar{\cD}_i$ is obtained by first sampling $v_i$ from $\cD_i$
and then rounding down to $\bar{v}_i = \min\{v_i , \hat{v}_i\}$).
Finally, denote by $\bar{\cI} = (N,M,\bar{\cD})$ the imaginary Bayesian instance
where players' values are drawn from $\bar{\cD}$.

Note that $\cD'$ is also a discretization distribution for $\bar{\cD}$,
following the proof and notations of Theorem \ref{thm:single},
letting $v_i^-$ be the value first sampled from $\bar{\cD_i}$ then rounding down to the support of $\cD'$,
we have $\cM_{EMR}$ is truthful and using the technique of Mechanism \ref{alg:imaginary},
we have
\begin{eqnarray}\label{eq:regular}
&&Rev(\cM_{EMR}(\cI)) = Rev(\cM_{MRS}(v, \cD')) \geq Rev(\cM_{MRS}(v', \cD')) \nonumber\\
&\geq& \sum\limits_{i}
\mathop\mathbb{E}\limits_{\bar{v}_{-i}\sim \bar{\cD}_{-i}}
p_{i}(\bar{v}_{-i}; \bar{\cD}) \cdot
\Pr_{\bar{v}_i\sim \bar{\cD}_i}[v^-_{i}\geq p_{i}(\bar{v}_{-i}; \bar{\cD})] \nonumber\\
&=& \sum\limits_{i}
\mathop\mathbb{E}\limits_{\bar{v}_{-i}\sim \bar{\cD}_{-i}}
p_{i}(\bar{v}_{-i}; \bar{\cD}) \cdot
\Pr_{\bar{v}_i\sim \bar{\cD}_i}[v^-_{i}\geq p_{i}(\bar{v}_{-i}; \bar{\cD})]
\cdot ({\bf I}_{v_i^* \leq \frac{16OPT(\cI)}{\epsilon}}
+ {\bf I}_{v_i^* > \frac{16OPT(\cI)}{\epsilon}}).
\end{eqnarray}

We bound the indicators separately.

\begin{eqnarray}\label{eq:low_tail}
&&\sum\limits_{i}
\mathop\mathbb{E}\limits_{\bar{v}_{-i}\sim \bar{\cD}_{-i}}
p_{i}(\bar{v}_{-i}; \bar{\cD}) \cdot
\Pr_{\bar{v}_i\sim \bar{\cD}_i}[v^-_{i}\geq p_{i}(\bar{v}_{-i}; \bar{\cD})]
\cdot {\bf I}_{v_i^* \leq \frac{16OPT(\cI)}{\epsilon}} \nonumber\\
&=& \sum\limits_{i}
\mathop\mathbb{E}\limits_{\bar{v}_{-i}\sim \bar{\cD}_{-i}}
p_{i}(\bar{v}_{-i}; \bar{\cD}) \cdot
\Pr_{\bar{v}_i\sim \bar{\cD}_i}[v^-_{i}\geq p_{i}(\bar{v}_{-i}; \bar{\cD})] \nonumber\\
&& \quad\quad\quad\quad\quad\quad\quad\quad\quad\quad\quad\quad\quad\quad\quad\quad
\cdot {\bf I}_{v_i^* \leq \frac{16OPT(\cI)}{\epsilon}}
\cdot ({\bf I}_{p_{i}(\bar{v}_{-i}; \bar{\cD}) < v_i^*}
+ {\bf I}_{p_{i}(\bar{v}_{-i}; \bar{\cD}) \geq v_i^*}) \nonumber\\
&\geq& \sum\limits_{i}
\mathop\mathbb{E}\limits_{\bar{v}_{-i}\sim \bar{\cD}_{-i}}
[ p_{i}(\bar{v}_{-i}; \bar{\cD})
\cdot \frac{1}{1+\frac{\epsilon}{4}} \cdot
\Pr_{\bar{v}_i\sim \bar{\cD}_i}[\bar{v}_{i}\geq p_{i}(\bar{v}_{-i}; \bar{\cD})]
\cdot {\bf I}_{v_i^* \leq \frac{16OPT(\cI)}{\epsilon}}
\cdot {\bf I}_{p_{i}(\bar{v}_{-i}; \bar{\cD}) < v_i^*} \nonumber\\
&&+ (p_{i}(\bar{v}_{-i}; \bar{\cD}) \cdot
\Pr_{\bar{v}_i\sim \bar{\cD}_i}[\bar{v}_{i}\geq p_{i}(\bar{v}_{-i}; \bar{\cD})]
- \frac{16OPT(\cI)}{\epsilon} \cdot \frac{\epsilon^2}{256n})
\cdot {\bf I}_{v_i^* \leq \frac{16OPT(\cI)}{\epsilon}}
\cdot {\bf I}_{p_{i}(\bar{v}_{-i}; \bar{\cD}) \geq v_i^*} ] \nonumber\\
&\geq& \frac{1}{1+\frac{\epsilon}{4}} \sum\limits_{i}
\mathop\mathbb{E}\limits_{\bar{v}_{-i}\sim \bar{\cD}_{-i}}
p_{i}(\bar{v}_{-i}; \bar{\cD}) \cdot
\Pr_{\bar{v}_i\sim \bar{\cD}_i}[\bar{v}_{i}\geq p_{i}(\bar{v}_{-i}; \bar{\cD})]
\cdot {\bf I}_{v_i^* \leq \frac{OPT(\cI)}{16\epsilon}}
- \frac{\epsilon}{16} \cdot OPT(\cI).
\end{eqnarray}

The first inequality here holds because for price $p_i(\bar{v}_{-i}; \bar{\cD}) < v_i^*$, we have
$$\Pr_{\bar{v}_i\sim \bar{\cD}_i}[v^-_{i}\geq p_{i}(\bar{v}_{-i}; \bar{\cD})]
\geq \frac{1}{1+\frac{\epsilon}{4}}
\Pr_{\bar{v}_i\sim \bar{\cD}_i}[\bar{v}_{i}\geq p_{i}(\bar{v}_{-i}; \bar{\cD})]$$
due to the structure of the quantile queries for $\cD'$.
For price $p_i(\bar{v}_{-i}; \bar{\cD}) \geq v_i^*$, when $v_i^* \leq \frac{16OPT(\cI)}{\epsilon}$,
by the regularity of $\cD_i$, the optimal reserve corresponds to the quantile interval $(\frac{\epsilon^2}{256n}, 1]$.
Thus we have
\begin{eqnarray*}
&& p_{i}(\bar{v}_{-i}; \bar{\cD}) \cdot
\Pr_{\bar{v}_i\sim \bar{\cD}_i}[v^-_{i}\geq p_{i}(\bar{v}_{-i}; \bar{\cD})] \geq 0\\
&\geq& p_{i}(\bar{v}_{-i}; \bar{\cD}) \cdot
\Pr_{\bar{v}_i\sim \bar{\cD}_i}[\bar{v}_{i}\geq p_{i}(\bar{v}_{-i}; \bar{\cD})]
- v_i^* \cdot \Pr_{\bar{v}_i\sim \bar{\cD}_i}[\bar{v}_{i}\geq v_i^*] \\
&\geq& p_{i}(\bar{v}_{-i}; \bar{\cD}) \cdot
\Pr_{\bar{v}_i\sim \bar{\cD}_i}[\bar{v}_{i}\geq p_{i}(\bar{v}_{-i}; \bar{\cD})]
- \frac{16OPT(\cI)}{\epsilon} \cdot \frac{\epsilon^2}{256n}
\end{eqnarray*}
since the expected revenue is non-decreasing for quantile range $[0, \frac{\epsilon^2}{256n}]$. 
Thus Equation \ref{eq:low_tail} holds.
Then for the second indicator for Equation \ref{eq:regular}, we have

\begin{eqnarray}\label{eq:high_tail}
&&\sum\limits_{i}
\mathop\mathbb{E}\limits_{\bar{v}_{-i}\sim \bar{\cD}_{-i}}
p_{i}(\bar{v}_{-i}; \bar{\cD}) \cdot
\Pr_{\bar{v}_i\sim \bar{\cD}_i}[v^-_{i}\geq p_{i}(\bar{v}_{-i}; \bar{\cD})]
\cdot {\bf I}_{v_i^* > \frac{16OPT(\cI)}{\epsilon}} \nonumber\\
&\geq& \frac{1}{1+\frac{\epsilon}{4}} \sum\limits_{i}
\mathop\mathbb{E}\limits_{\bar{v}_{-i}\sim \bar{\cD}_{-i}}
p_{i}(\bar{v}_{-i}; \bar{\cD}) \cdot
\Pr_{\bar{v}_i\sim \bar{\cD}_i}[\bar{v}_{i}\geq p_{i}(\bar{v}_{-i}; \bar{\cD})]
\cdot {\bf I}_{v_i^* > \frac{16OPT(\cI)}{\epsilon}}
\end{eqnarray}
also by the construction the quantile queries for $\cD'$.
Combining Equation \ref{eq:regular}, \ref{eq:low_tail} and \ref{eq:high_tail}, we have
\begin{eqnarray*}
&&Rev(\cM_{EMR}(\cI)) \nonumber\\
&\geq& \frac{1}{1+\frac{\epsilon}{4}} \sum\limits_{i}
\mathop\mathbb{E}\limits_{\bar{v}_{-i}\sim \bar{\cD}_{-i}}
p_{i}(\bar{v}_{-i}; \bar{\cD}) \cdot
\Pr_{\bar{v}_i\sim \bar{\cD}_i}[\bar{v}_{i}\geq p_{i}(\bar{v}_{-i}; \bar{\cD})]
\cdot {\bf I}_{v_i^* \leq \frac{16OPT(\cI)}{\epsilon}}
- \frac{\epsilon}{16} \cdot OPT(\cI) \nonumber\\
&& + \frac{1}{1+\frac{\epsilon}{4}} \sum\limits_{i}
\mathop\mathbb{E}\limits_{\bar{v}_{-i}\sim \bar{\cD}_{-i}}
p_{i}(\bar{v}_{-i}; \bar{\cD}) \cdot
\Pr_{\bar{v}_i\sim \bar{\cD}_i}[\bar{v}_{i}\geq p_{i}(\bar{v}_{-i}; \bar{\cD})]
\cdot {\bf I}_{v_i^* > \frac{16OPT(\cI)}{\epsilon}} \nonumber\\
&=& \frac{1}{1+\frac{\epsilon}{4}} Rev(\cM_{MRS}(\bar{v},\bar{\cD}))
- \frac{\epsilon}{16} \cdot OPT(\cI)
\end{eqnarray*}

By Lemma 2 of \cite{devanur2016sample}, for $0\leq \epsilon \leq 1$,
$Rev(\cM_{MRS}(\bar{v},\bar{\cD})) = OPT(\bar{\cI}) \geq (1-\frac{\epsilon}{4})OPT(\cI)$.
Thus we have
\begin{eqnarray*}
&&Rev(\cM_{EMR}(\cI)) \nonumber\\
&\geq& \frac{1}{1+\frac{\epsilon}{4}} (1-\frac{\epsilon}{4})OPT(\cI) - \frac{\epsilon}{16} \cdot OPT(\cI)
\geq \frac{1}{1+\epsilon} OPT(\cI).
\end{eqnarray*}
Thus Theorem \ref{thm:single:regular} holds.
\end{proof}

Following \cite{cole2014sample, huang2015making, devanur2016sample},
the sample complexity for single-item auction with regular distributions is bounded between
$\Omega(\max\{n\epsilon^{-1}, \epsilon^{-3}\})$ and $\tilde{O}(n \epsilon^{-4})$.
However, each sample is a valuation profile of the players, and thus contains $n$ values.
When $\epsilon$ is small, the query complexity in this setting is $O(n\epsilon^{-1}\log\frac{n}{\epsilon})$.
Thus the query complexity is still much lower than the sample complexity.

\section{Applications: Sampling Mechanisms}
\label{sec:app:sample}

Using our techniques for query complexity, we can easily construct sampling mechanisms for multi-parameter auctions.
Currently, the sample complexity for unit-demand auctions and additive distributions has been upper-bounded
in \cite{azar2014prophet,morgenstern2016learning,goldner2016prior,cai2017learning}
for bounded auctions.
In this section, we provide another way to explicitly construct sampling mechanisms for both unit-demand and additive auctions, for arbitrary distributions with
small-tails (and for bounded distributions).

The idea is to use samples to approximate {\em quantile} queries.
Mechanism \ref{mech:sample-to-query} defines our sampling mechanism $\cM_{SM}$.
Recall that mechanisms $\cM_{MRS}$, $\cM_{UD}$ and $\cM_{A}$ are known (approximately) optimal DSIC mechanisms for
single-item, unit-demand and additive auctions respectively.
Note that in mechanism $\cM_{SM}$,
we use a different method to discretize the quantile space
for additive auctions, so as to further reduce its sample complexity.
In particular, we have the following theorem. 

\begin{algorithm}[htbp]
\floatname{algorithm}{Mechanism}
  \caption{\hspace{-3pt} Sampling Mechanism $\cM_{SM}$ } 
  \label{mech:sample-to-query}
    \begin{algorithmic}[1]
\STATE For single-item auctions and unit-demand auctions, given $\epsilon > 0$, set $\delta = \frac{\epsilon}{6}$,
  $\epsilon_1 = h(\frac{2\epsilon}{3(1+\epsilon)})$ and
  $k = \lceil \log_{1+\delta}\frac{1}{\epsilon_1}\rceil$;
  define the {\em quantile vector} as
  $q=(q_{0},q_{1},\dots, q_{k-1}, q_{k})
  =(1,\epsilon_1(1+\delta)^{k-1}, \dots, \epsilon_1(1+\delta), \epsilon_1)$.

  For additive auctions, given $\epsilon>0$, set 
  $\epsilon_1 = h(\frac{\epsilon}{10(1+\epsilon)})$
  and $k = \lfloor \frac{1}{\epsilon_1}\rfloor$;
  define the {\em quantile vector} as
  $q=(q_{0},q_{1},\dots, q_{k-1}, q_{k})=
  (1,k\epsilon_1, \dots, 2\epsilon_1, \epsilon_1)$.
  \\

\STATE For each player $i$ and item $j$, given $t$ samples $V_{ij}^t = \{v_{ij}^1, \dots, v_{ij}^t\}$, without loss of generality assume $v_{ij}^1\geq v_{ij}^2\geq \cdots\geq v_{ij}^t$.
For each quantile $q_l$, set $v_{ij}^{t q_l}$ to be the value
corresponding to the quantile query $q_l$. (If $tq_l$ is not an integer then the mechanism takes $\lceil tq_l\rceil$.)\\

\STATE Construct a discrete distribution $\cD'_{ij}$ as follows:
$\cD'_{ij}(v_{ij}^{tq_l}) = q_l - q_{l+1}$ for each $l\in \{0, \dots, k-1\}$,
and $\cD'_{ij}(v_{ij}^{tq_k}) = \epsilon_1$.
Finally, let $\cD'_i = \times_{j\in M} \cD'_{ij}$ for each player $i$ and
let $\cD'=\times_{i\in N} \cD'_i$. \\

\STATE Run $\cM_{MRS}$/$\cM_{UD}$/$\cM_{A}$ with distribution $\cD'$ and the players' reported values.
  \end{algorithmic}
\end{algorithm}

\begin{theorem}\label{thm:sample-to-query}
$\forall \epsilon > 0$ and $\gamma\in(0,1)$, for any Bayesian instance $\cI = (N,M,\cD)$,
\begin{itemize}
\item for single-item auctions satisfying the Small-Tail Assumption 2,
         with $\tilde{O}(h^{-2}(\frac{2\epsilon}{3(1+\epsilon)}) \cdot (\frac{\epsilon}{1+\epsilon})^{-2})$ samples,
	mechanism $\cM_{SM}$ achieves revenue at least $\frac{1}{1+\epsilon} OPT(\cI)$ with probability at least $1-\gamma$;
\item for unit-demand auctions  satisfying the Small-Tail Assumption 2,
         with $\tilde{O}(h^{-2}(\frac{2\epsilon}{3(1+\epsilon)}) \cdot (\frac{\epsilon}{1+\epsilon})^{-2})$ samples,
	mechanism $\cM_{SM}$ achieves revenue at least $\frac{1}{24(1+\epsilon)} OPT$ with probability at least $1-\gamma$;
\item for additive auctions satisfying the Small-Tail Assumption 1,
         with $\tilde{O}(h^{-2}(\frac{\epsilon}{10(1+\epsilon)})
(\frac{1}{2}-\frac{1}{1+(1+\frac{\epsilon}{5})^{1/m}}))^{-2})$ samples,
	mechanism $\cM_{SM}$ achieves revenue at least $\frac{1}{8(1+\epsilon)} OPT$ with probability at least $1-\gamma$.
\end{itemize}
\end{theorem}

\begin{proof}

After constructing the distributions, we simply run the existing DSIC mechanisms as a Blackbox,
and if the constructed distribution satisfies the property that for any quantile $q_l$,
\begin{equation}\label{eq:quantile_ratio}
q_{ij}(v_{ij}^{t\cdot q_{l+1}})
\geq \frac{1}{1+\frac{\epsilon}{3}}
\left(q_{ij}(v_{ij}^{t\cdot q_l}) \right).
\end{equation}
all our query complexity results for single-item and unit-demand auctions directly apply here.

Since here for sampling mechanism, we slice the quantile interval uniformly,
in the ideal case, the selected sampled values correspond to the desired quantiles and
$\cD_{ij}(v_{ij}^{t\cdot q_l}) = \cD_{ij}(v_{ij}^{t\cdot q_{l+1}})$.
However, since these samples are random, we may not obtain the ideal case.
In fact, given parameter $d = \frac{12+3\epsilon}{\epsilon}$,
if for any quantile $q_l$,
\begin{equation}\label{eq:sample:quantile}
q_l - \frac{q_l}{d}  \leq
q_{ij}(v_{ij}^{t\cdot q_l})  \leq
q_l + \frac{q_l}{d},
\end{equation}
then
$$
\dfrac{q_{ij}(v_{ij}^{t\cdot q_{l+1}})}
{q_{ij}(v_{ij}^{t\cdot q_l})}
\geq \dfrac{q_{l+1}(1-\frac{1}{d})}
{q_l(1+\frac{1}{d}) }
\geq \dfrac{\frac{1}{1+\frac{\epsilon}{6}}
(1-\frac{1}{d})}{1+\frac{1}{d} }
= \dfrac{1}{1+\frac{\epsilon}{3}}
,
$$
for any $\epsilon > 0$, that is, Equantion \ref{eq:quantile_ratio} holds.
In the following, we show how many samples are enough to obtain Inequality \ref{eq:sample:quantile}.

First, we bound the probability that $v_{ij}^{t\cdot q_{l}}$ locates in the quantile interval $[q_l - \frac{q_l}{d}, q_l + \frac{q_l}{d}]$.
Let $E_{ij,l}^{left}$ be the event that $v_{ij}^{t\cdot q_{l}}$ locates in the quantile interval $[0, q_l - \frac{q_l}{d}]$,
and $E_{ij,l}^{right}$ be the event that $v_{ij}^{t\cdot q_{l}}$ locates in the quantile interval $[q_l + \frac{q_l}{d}, 1]$.
Then
$$\Pr[E_{ij,l}^{left}] =
\sum_{s=0}^{t-t\cdot q_{l}} \binom{t}{s} \left(q_{l}-\frac{q_l}{d}\right)^{s}
\left(1-q_{l}+\frac{q_l}{d}\right)^{t-s},
$$
and
$$\Pr[E_{ij,l}^{right}] =
\sum_{s=0}^{t\cdot q_{l}} \binom{t}{s} \left(1-q_{l}-\frac{q_l}{d}\right)^{s}
\left(q_{l}+\frac{q_l}{d}\right)^{t-s}.
$$
By Chernoff's inequality and $\forall i,j,l$, letting $\Pr[E_{ij,l}^{left}]$ and $\Pr[E_{ij,l}^{right}]$ be no more than $\frac{\gamma}{2mn(k+1)}$,
$t=\tilde{O}((\frac{\epsilon_1}{d})^{-2})
= \tilde{O}((\frac{\epsilon \cdot \epsilon_1}{(1+\epsilon)})^{-2})$.
That is with $\tilde{O}(h^{-2}(\frac{2\epsilon}{3(1+\epsilon)}) \cdot (\frac{\epsilon}{1+\epsilon})^{-2})$ samples,
the probability that $v_{ij}^{t\cdot q_{l}}$ does not locate in the quantile interval
$[q_l - \frac{q_l}{d}, q_l + \frac{q_l}{d}]$
is less than $\frac{\gamma}{mn(k+1)}$.
By union bound, there exists one $v_{ij}^{t\cdot q_{l}}$ for all
$i\in [n], j\in [m], l\in [k+1]$
does not locate in the quantile interval
$[q_l - \frac{q_l}{d}, q_l + \frac{q_l}{d}]$
is less than $\gamma$.
Then with probability $1-\gamma$, Inequality \ref{eq:sample:quantile} holds.

For additive auctions,
if the constructed distribution satisfies the property that for any quantile~$q_l$,
\begin{equation}\label{eq:add:quantile_ratio}
q_{ij}(v_{ij}^{t\cdot q_l}) - q_{ij}(v_{ij}^{t\cdot q_{l+1}})
\geq \frac{1}{(1+\frac{\epsilon}{5})^{1/m}}
\left(q_{ij}(v_{ij}^{t\cdot q_{l+1}}) - q_{ij}(v_{ij}^{t\cdot q_{l+2}})\right),
\end{equation}
all our query complexity results for additive auctions directly apply here.
In fact, if for any quantile~$q_l$,
\begin{equation*}
q_l - \epsilon_1(\frac{1}{2}-\frac{1}{1+(1+\frac{\epsilon}{5})^{1/m}})
\leq q_{ij}(v_{ij}^{t\cdot q_l})
\leq q_l + \epsilon_1(\frac{1}{2}-\frac{1}{1+(1+\frac{\epsilon}{5})^{1/m}}),
\end{equation*}
then,
$$
\dfrac{q_{ij}(v_{ij}^{t\cdot q_l})  - q_{ij}(v_{ij}^{t\cdot q_{l+1}})}
{q_{ij}(v_{ij}^{t\cdot q_{l+q}}) - q_{ij}(v_{ij}^{t\cdot q_{l+2}})}
\geq \dfrac{\epsilon_{1} - \epsilon_1(1-\frac{2}{1+(1+\frac{\epsilon}{5})^{1/m}})}
{\epsilon_{1} + \epsilon_1(1-\frac{2}{1+(1+\frac{\epsilon}{5})^{1/m}})}
= \dfrac{1}{(1+\frac{\epsilon}{5})^{1/m}},
$$

Using the same technique of applying the Chernoff's inequality,
with $\tilde{O}(h^{-2}(\frac{\epsilon}{10(1+\epsilon)})
(\frac{1}{2}-\frac{1}{1+(1+\frac{\epsilon}{5})^{1/m}}))^{-2})$ samples,
Equation \ref{eq:add:quantile_ratio} holds with probability $1-\gamma$.
Thus Theorem \ref{thm:sample-to-query} holds.
\end{proof}

\subparagraph*{Remark.}
Following the convention in the literature, a logarithmic factor depending on $\gamma$ has been absorbed
in $\tilde{O}(\cdot)$.
If the values are bounded in $[1,H]$,
by defining the tail function $h$ according to $H$, the resulting
sample complexity is $\tilde{O}(m^{4}n^{2}H^{2}(1+\epsilon)^4\epsilon^{-4})$ for unit-demand auctions
and
$\tilde{O}(m^{4}n^{2}H^{2}(\frac{1+\epsilon}{\epsilon})^2
(\frac{1}{2}-\frac{1}{1+(1+\frac{\epsilon}{5})^{1/m}}))^{-2})$
for additive auctions.

We note that our resulting sampling mechanisms are not better than the best-known sample complexity results
as shown in Table \ref{table:sample:laterature}.
However, our mechanisms are able to deal with unbounded distributions
satisfy small-tail assumptions.

\newpage

\appendix

\section{Missing Proofs for Section \ref{sec:queryunbounded}}
\label{app:unbounded}

We now prove the claims used above.

\smallskip

\noindent
{\bf Claim \ref{clm:rev:gap}} (restated){\bf.} {\em
$\Pr_{v_i\sim \cD_i}[v_{i}\geq p_{i}(v_{-i}; \cD) | {q_{i}(v_{i}) > \epsilon_1}]
\leq (1+\delta) \Pr_{v_i\sim \cD_i}[v^-_{i}\geq p_{i}(v_{-i}; \cD)].$}

\begin{proof}

By definition, ${q_{i}(v_{i}) > \epsilon_1}$ implies $v_i \leq v'_{i;k}$,
where $v'_{i;k}$ is the largest value in $V'_i$, the support of distribution $\cD'_i$.
Note that $v_i^-\leq v'_{i;k}$ for any $v_i$.
If $p_{i}(v_{-i}; \cD) > v'_{i;k}$, then both probabilities are 0 and the inequality holds.

Below we consider the case $p_{i}(v_{-i}; \cD) \leq v'_{i;k}$.
Let
$v'_{i;-1} = -1$ and
$l\in \{0, 1, \dots, k\}$ be such that $v'_{i;l} \geq p_i(v_{-i}; \cD)$ and $v'_{i; l-1} < p_i(v_{-i}; \cD)$.
We have

\begin{eqnarray*}
&&\Pr_{v_i\sim \cD_i}[v_{i}\geq p_{i}(v_{-i}; \cD)  | q_{i}(v_{i}) > \epsilon_1] \\
&\leq& \Pr_{v_i\sim \cD_i}[v_{i}^+ \geq p_i(v_{-i}; \cD)  | q_{i}(v_{i}) > \epsilon_1] \\
&=& \Pr_{v_i\sim \cD_i}[v_{i}^+ \geq v'_{i;l}  | q_{i}(v_{i}) > \epsilon_1] \\
&=& \Pr_{v_i\sim \cD_i}[v_{i}^- \geq v'_{i;l-1}  | q_{i}(v_{i}) > \epsilon_1] \\
&=& \Pr_{v_i\sim \cD_i}[v_{i}^- \geq v'_{i;\max\{0, l-1\}}  | q_{i}(v_{i}) > \epsilon_1] \\
&= & \frac{\Pr_{v_i\sim \cD_i}[v_{i}^- \geq v'_{i;\max\{0, l-1\}} \mbox{ and } q_{i}(v_{i}) > \epsilon_1]}{\Pr_{v_i\sim \cD_i}[q_{i}(v_{i}) > \epsilon_1]} \\
&=& \frac{\Pr_{v_i\sim \cD_i}[v_{i}^- \geq v'_{i;\max\{0, l-1\}}] - \Pr_{v_i\sim \cD_i}[v_{i}^- \geq v'_{i;\max\{0, l-1\}} \mbox{ and } q_{i}(v_{i}) \leq \epsilon_1]}
{\Pr_{v_i\sim \cD_i}[q_{i}(v_{i}) > \epsilon_1]} \\
&=& \frac{\Pr_{v_i\sim \cD_i}[v_{i}^- \geq v'_{i;\max\{0, l-1\}}] - \Pr_{v_i\sim \cD_i}[q_{i}(v_{i}) \leq \epsilon_1]}
{\Pr_{v_i\sim \cD_i}[q_{i}(v_{i}) > \epsilon_1]} \\
&\leq& \Pr_{v_i\sim \cD_i}[v_{i}^- \geq v'_{i;\max\{0, l-1\}}] \\
& = & \Pr_{v'_i\sim \cD'_i}[v'_i \geq v'_{i;\max\{0, l-1\}}] =  q_{\max\{0, l-1\}} \leq (1+\delta) q_l \\
& = & (1+\delta) \Pr_{v'_i\sim \cD'_i}[v'_i \geq v'_{i;l}] =  (1+\delta) \Pr_{v'_i\sim \cD'_i}[v'_i \geq p_{i}(v_{-i}; \cD)] \\
&=& (1+\delta) \Pr_{v_i\sim \cD_i}[v_{i}^-\geq p_{i}(v_{-i}; \cD)],
\end{eqnarray*}

\noindent
as desired.
Indeed, the first inequality is because $v_i^+>v_i$,
and
the first equality is because $v_i^+\in V'_i\cup\{+\infty\}$ and thus
$v_i^+\geq p_i(v_{-i}; \cD)$ if and only if $v_i^+\geq v'_{i;l}$.
Similarly, the second equality is because $(v_i^-, v_i^+)$ and $(v'_{i;l-1}, v'_{i;l})$
are two
pairs of consecutive values in $V'_i\cup\{-1, +\infty\}$,
thus $v_i^+\geq v'_{i;l}$ if and only if $v_i^-\geq v'_{i;l-1}$.
The third
equality is because $v_i^-\geq v'_{i;0}$ always.
The sixth equality is because $q_{i}(v_{i}) \leq \epsilon_1$ implies $v_i\geq v'_{i;k}\geq v'_{i;l}$, thus $v_{i}^- \geq v'_{i;\max\{0, l-1\}}$.
The seventh equality is by the definition of the round-down scheme.
The following two equalities  and the inequality are by the construction of $\cD'_i$ and the definition of the quantile vector $q$.
Indeed, $(1+\delta)q_0 =1+\delta>1=q_0$,
$(1+\delta)q_1 = \epsilon_1(1+\delta)^k\geq
\epsilon_1 (1+\delta)^{\log_{1+\delta}\frac{1}{\epsilon_1}}
= \epsilon_1 \cdot \frac{1}{\epsilon_1} = 1 = q_0$,
and
$(1+\delta)q_l = q_{l-1}$ for any $l\geq 2$.
The second-last equality is because $v'_i\in V'_i$, thus $v'_i \geq v'_{i;l}$ if and only if $v'_i\geq p_{i}(v_{-i}; \cD)$.
Finally, the last equality is again by the definition of the round-down scheme.
\end{proof}

\noindent
{\bf Claim \ref{clm:rev:m*}} (restated){\bf.} {\em
$Rev(\cM^*(\cI')) \geq \frac{1}{1+\epsilon}OPT(\cI)$.}
\begin{proof}
Combining Equation \ref{eqn:3} and Claim \ref{clm:rev:gap}, we have
$$Rev(\cM^{*}(\cI'))
\geq \frac{1}{1+\delta}\sum\limits_{i}
\mathop\mathbb{E}\limits_{v_{-i}\sim \cD_{-i}}
p_{i}(v_{-i}; \cD) 
\cdot \Pr_{v_i\sim \cD_i}[v_{i}\geq p_{i}(v_{-i}; \cD) | {q_{i}(v_{i}) > \epsilon_1}].$$
Accordingly,
\begin{eqnarray}\label{equ:M*OPTI}
&& Rev(\cM^{*}(\cI'))
\geq \frac{1}{1+\delta}\sum\limits_{i}
\mathop\mathbb{E}\limits_{v_{-i}\sim \cD_{-i}}
p_{i}(v_{-i}; \cD) 
\cdot \Pr_{v_i\sim \cD_i}[v_{i}\geq p_{i}(v_{-i}; \cD) | {q_{i}(v_{i}) > \epsilon_1}] \nonumber \\
&\geq & \frac{1}{1+\delta}\sum\limits_{i}
\mathop\mathbb{E}\limits_{v_{-i}\sim \cD_{-i}}
p_{i}(v_{-i}; \cD) \cdot
\Pr_{v_{i}\sim \cD_{i}}[q_{i}(v_{i}) > \epsilon_1 \mbox{ and } v_{i}\geq p_{i}(v_{-i}; \cD)] \nonumber \\
&=& \frac{1}{1+\delta} \sum\limits_{i}
\mathop\mathbb{E}\limits_{v_{-i}\sim \cD_{-i}}
\mathop\mathbb{E}\limits_{v_{i}\sim \cD_{i}}
p_{i}(v_{-i}; \cD) \cdot {\bf I}_{q_{i}(v_{i}) > \epsilon_1} \cdot {\bf I}_{v_{i}\geq p_{i}(v_{-i}; \cD)} \nonumber \\
&=& \frac{1}{1+\delta}
\mathop\mathbb{E}\limits_{v\sim \cD}
\sum\limits_{i} p_{i}(v_{-i}; \cD) \cdot {\bf I}_{q_{i}(v_{i}) > \epsilon_1} \cdot {\bf I}_{v_{i}\geq p_{i}(v_{-i}; \cD)} \nonumber \\
&\geq& \frac{1}{1+\delta}
\mathop\mathbb{E}\limits_{v\sim \cD} {\bf I}_{\forall i, q_{i}(v_{i}) > \epsilon_1} \cdot
\sum\limits_{i} p_{i}(v_{-i}; \cD) {\bf I}_{v_{i}\geq p_{i}(v_{-i}; \cD)} \nonumber \\
&=& \frac{1}{1+\delta}
\mathop\mathbb{E}\limits_{v\sim \cD} {\bf I}_{\forall i, q_{i}(v_{i}) > \epsilon_1} \cdot Rev_{OPT}(v;\cI) \nonumber \\
&\geq&
\frac{1-\delta_1}{1+\delta}OPT(\cI).
\end{eqnarray}

\noindent
Here the second last equality holds by the definition of $p_{i}(v_{-i}; \cD)$ and $Rev_{OPT}(v;\cI)$,
and last inequality holds by the Small-Tail Assumption 2.
Since $\delta = \frac{\epsilon}{3}$ and $\delta_1 = \frac{2\epsilon}{3(1+\epsilon)}$, we have
$$\frac{1-\delta_1}{1+\delta}
= \frac{1}{1+\epsilon},$$
thus Claim \ref{clm:rev:m*} holds.
\end{proof}

\newpage

\bibliographystyle{abbrv}
\bibliography{ref}

\end{document}